\documentclass{article}

\usepackage[T1]{fontenc}
\usepackage[margin=1in]{geometry}
\usepackage{amsfonts,amsmath,amsthm,amssymb,mathtools}  %

\mathtoolsset{centercolon}
\usepackage{xfrac,nicefrac}
\usepackage{mathdots}
\usepackage{mleftright}  %

\usepackage{xspace}
\xspaceaddexceptions{]\}}  %
\usepackage{regexpatch}

\usepackage{bm,bbm,dsfont}  %
\usepackage{caption}
\usepackage[normalem]{ulem}
\usepackage{enumitem}

\usepackage{graphicx}
\usepackage{float}
\usepackage{subcaption}  %
\usepackage{tcolorbox}
\usepackage{tikz}
\usetikzlibrary{decorations.pathreplacing}
\usetikzlibrary{calc}
\usetikzlibrary{positioning}
\usetikzlibrary{arrows.meta}

\usepackage[linesnumbered,boxed,ruled,vlined]{algorithm2e}

\SetCommentSty{mycommfont}

\usepackage{thmtools,thm-restate}
\usepackage[colorlinks,citecolor=blue,linkcolor=blue,urlcolor=red]{hyperref}

\theoremstyle{plain}

\newtheorem{theorem}{Theorem}[section]  %
\newtheorem{lemma}[theorem]{Lemma}

\newtheorem{corollary}[theorem]{Corollary}

\newtheorem{claim}[theorem]{Claim}

\theoremstyle{definition}  %

\newtheorem{property}[theorem]{Property}

\usepackage[capitalise]{cleveref}
\crefname{algocf}{Algorithm}{Algorithms}
\Crefname{algocf}{Algorithm}{Algorithms}
\crefname{claim}{Claim}{Claims}
\Crefname{claim}{Claim}{Claims}
\crefname{property}{Property}{Properties}
\Crefname{property}{Property}{Properties}

\usepackage{xurl}

\newfloat{Distribution}{htbp}{loa}
\crefname{Distribution}{Distribution}{Distributions}
\Crefname{Distribution}{Distribution}{Distributions}
\newfloat{Protocol}{htbp}{loa}
\crefname{Protocol}{Protocol}{Protocols}
\Crefname{Protocol}{Protocol}{Protocols}

\SetKwProg{Function}{Function}{:}{}

\SetKwProg{CodeBlock}{}{}{}
\SetKwProg{Repeat}{Repeat}{:}{}

\DeclarePairedDelimiter{\ceil}{\lceil}{\rceil}
\DeclarePairedDelimiter{\floor}{\lfloor}{\rfloor}

\DeclarePairedDelimiter{\bk}{(}{)}
\DeclarePairedDelimiter{\Bk}{[}{]}
\DeclarePairedDelimiter{\BK}{\{}{\}}

\DeclarePairedDelimiter{\angbk}{\langle}{\rangle}

\DeclarePairedDelimiterX\mysetbase[2]{\lbrace}{\rbrace}{#1\,\delimsize\vert\,#2}
\NewDocumentCommand{\myset}{sO{}m m}{%
  \IfBooleanTF{#1}%
    {\mysetbase*{#3}{#4}}%
    {\mysetbase[#2]{#3}{#4}}%
}

\DeclareMathOperator*{\E}{\mathbb{E}}

\let\Pr\PrAux
\DeclareMathOperator{\poly}{poly}

\DeclareMathOperator*{\ind}{\mathbbm{1}}

\newcommand{\F}{\mathbb{F}}

\renewcommand{\tilde}{\widetilde}

\newcommand{\defeq}{\coloneqq}
\newcommand{\eps}{\varepsilon}

\renewcommand{\epsilon}{\eps}

\newcommand{\defn}[1]{{\boldmath\textbf{#1}}}

\newcommand{\numberthis}{\addtocounter{equation}{1}\tag{\theequation}}

\usepackage{regexpatch}
\makeatletter
\xpatchcmd\thmt@restatable{%
\csname #2\@xa\endcsname\ifx\@nx#1\@nx\else[{#1}]\fi
}{%
\ifthmt@thisistheone
\csname #2\@xa\endcsname\ifx\@nx#1\@nx\else[{#1}]\fi
\else
\csname #2\@xa\endcsname[{Restated}]
\fi}{}{}
\makeatother

\newcommand{\tperm}{t_{\text{perm}}}
\newcommand{\tcoup}{t_{\text{coupon}}}
\newcommand{\tslow}{t_{\text{slow}}}

\newcommand{\pb}{\textsf{Probe}}
\newcommand{\ft}{\textsf{Foot}}

\usepackage{todonotes}

\makeatletter
\let\oldparagraph\paragraph
\renewcommand{\paragraph}[1]{%
  \oldparagraph{\boldmath #1}%
}
\makeatother

\title{Static Retrieval Revisited: To Optimality and Beyond}
\newcommand{\authorspace}{5em}
\author{
  \hspace{\authorspace}
  Yang Hu\thanks{Institute for Interdisciplinary Information Sciences, Tsinghua University. \texttt{y-hu22@mails.tsinghua.edu.cn}.}
  \and
  William Kuszmaul\thanks{Carnegie Mellon University. Partially supported by NSF grant CNS-2504471. \texttt{kuszmaul@cmu.edu}.}
  \and
  Jingxun Liang\thanks{Carnegie Mellon University. \texttt{jingxunl@andrew.cmu.edu}.}
  \hspace{\authorspace}
  \and
  \hspace{\authorspace}
  Huacheng Yu\thanks{Princeton University. Supported by NSF CAREER award CCF-2339942. \texttt{yuhch123@gmail.com}.}
  \and
  Junkai Zhang\thanks{Institute for Interdisciplinary Information Sciences, Tsinghua University. \texttt{zhangjk22@mails.tsinghua.edu.cn}.}
  \and
  Renfei Zhou\thanks{Carnegie Mellon University. Partially supported by the Jane Street Graduate Research Fellowship and the MongoDB PhD Fellowship. \texttt{renfeiz@andrew.cmu.edu}.}
  \hspace{\authorspace}
}
\date{}

\begin{document}

\maketitle

\begin{abstract}
In the static retrieval problem, a data structure must answer retrieval queries mapping a set of $n$ keys in a universe $[U]$ to $v$-bit values. Information-theoretically, retrieval data structures can use as little as $nv$ bits of space. For small value sizes $v$, it is possible to achieve $O(1)$ query time while using space $nv + o(n)$ bits---whether or not such a result is possible for larger values of $v$ (e.g., $v = \Theta(\log n)$) has remained open. 

In this paper, we obtain a tight lower bound (as well as matching upper bounds) for the static retrieval problem. In the case where values are large, we show that there is actually a significant tension between time and space. It is not possible, for example, to get $O(1)$ query time using $nv + o(n)$ bits of space, when $v = \Theta(\log n)$ (and assuming the word RAM model with $O(\log n)$-bit words).

At first glance, our lower bound would seem to render retrieval unusable in many settings that aim to achieve very low redundancy. However, our second result offers a way around this: We show that, whenever a retrieval data structure $D_1$ is stored along with another data structure $D_2$ (whose size is similar to or larger than the size of $D_1$), it is possible to implement the combined data structure $D_1 \cup D_2$ so that queries to $D_1$ take $O(1)$ time, operations on $D_2$ take the same asymptotic time as if $D_2$ were stored on its own, and the total space is $nv + \mathrm{Space}(D_2) + n^{0.67}$ bits.
\end{abstract}

\section{Introduction}

The static \defn{retrieval} problem asks a data structure to store a function mapping a set of $n$ keys from some universe $[U]$ to $v$-bit values.\footnote{Although most works on the topic consider only integer $v$, one can also generalize the problem to allow for values from some arbitrary range $[V]$. Several of the results in this paper (the lower bounds for both retrieval and augmented retrieval, and the upper bounds for augmented retrieval) also apply to this setting.} The data structure must support queries such that, for any query key, which is \emph{guaranteed to be in the set}, the associated value can be efficiently retrieved (if the queried key is not in the set, the data structure is allowed to return anything).
Retrieval data structures are closely related to \emph{dictionaries}, except that they do not necessarily detect if a key is in the input set. 
In addition to being studied on their own \cite{dietzfelbinger2008succinct,porat2009optimal,dietzfelbinger2019constanttime,demaine2006dictionariis,mortensen2005dynamic,kuszmaul2024space,hagerup2001efficient,fredman1984size,mehlhorn1984data,kuszmaul2025tight}, retrieval data structures are widely used in both theory \cite{bender2025optimal, bender2023tiny,kuszmaul2025tight,porat2009optimal,dietzfelbinger2008succinct} and practice \cite{graf2020xor, dillinger2022fast} as a building block for constructing space-efficient data structures. 
Since retrieval data structures do not require storage of the keys, they need not use as much space as a dictionary---the information-theoretic space lower bound for a static retrieval data structure is $nv$ bits.

Much of the work on retrieval \cite{mehlhorn1984data,hagerup2001efficient,porat2009optimal,dietzfelbinger2019constanttime,dietzfelbinger2008succinct} focuses on the question of whether one can construct \emph{time-efficient} solutions that use space close to $nv$ bits. A solution using $nv + R$ bits is said to have \defn{redundancy} $R$. Using a minimal perfect hash function~\cite{mehlhorn1984data,hagerup2001efficient}, one can construct a retrieval data structure with constant query time and redundancy $O(n)$ bits. 
In the case where $v$ is small, several solutions have been proposed that achieve $O(1)$ query time with $o(n)$ bits of redundancy \cite{porat2009optimal, dietzfelbinger2019constanttime}. It has been widely believed that, even when $v = \Theta(\log n)$ (and using machine word size $w = \Theta(\log n)$), it should be possible to achieve $O(1)$ query time and $o(n)$ bits of redundancy---but such a construction has, so far, proven elusive.\footnote{Note that, for the closely related static dictionary problem \cite{yu2020nearly, hu2025optimal}, it is known how to construct $O(1)$-time solutions with very low redundancy, even when the keys/values involved are $\Theta(\log n)$ bits.}

There have also been several solutions proposed that exhibit time/space trade-offs \cite{dietzfelbinger2019constanttime, dietzfelbinger2008succinct}. An extreme point on the space-efficiency side is an elegant result by Dietzfelbinger and Walzer~\cite{dietzfelbinger2019constanttime}, who demonstrate (assuming free randomness and machine word size $w = \Theta(\log n)$) that $O(\log n)$ bits of redundancy and constant query time are possible when the values are restricted to a single bit (i.e., $v = 1$). This extends to large $v$ (e.g., $v = \Theta(\log n)$) by storing a separate data structure for each of the $v$ bits in the value. Such a solution achieves nearly optimal space, but also leads to a much slower query time of $\Theta(v)$, as it retrieves one bit at a time.

The first result of this paper is a cell-probe lower bound for the static retrieval problem. We find that, surprisingly, when $v$ is large (e.g., $v = \Theta(\log n)$), both the minimal perfect hashing solutions \cite{hagerup2001efficient} and the aforementioned solution of appending $v$ one-bit solutions together \cite{dietzfelbinger2019constanttime} sit at optimal points on the trade-off curve between time and space redundancy. This lower bound also extends to a smooth trade-off between value range, query time, and redundancy.

\begin{theorem}\label{thm:retrieval_ub_lb}
    Any static retrieval data structure that maps $n$ keys from a universe $[U]$ to $v$-bit values with query time $t$ must use $nv + \floor{n \cdot e^{-O(wt/v)}}$ bits of space in the cell-probe model with word size $w \ge v$.
    Moreover, if $2^v, U \le \poly(n)$, there is also a data structure (assuming free access to random bits) that matches this space bound up to an additive $O(\log^2 n)$ term.
\end{theorem}

When $v, w = \Theta(\log n)$, our lower bound matches the retrieval data structure using minimal perfect hashing for $t = O(1)$, and matches $O(\log n)$ copies of one-bit retrieval for $t=\Theta(\log n)$. We remark that the full version of the lower bound (Theorem \ref{thm:non_augmented_lb}) also applies to retrieval data structures with arbitrary value ranges $V$ (not necessarily powers of two). On the upper bound side, in addition to presenting upper bounds that assume free random bits, we also show that this assumption can be removed at the cost of an additional $\tilde{O}(n^{2/3})$ bits of redundancy (see Corollary~\ref{col:non_augmented_retrieval_remove_random}).

Our second result offers a surprising way around this lower bound.
Consider an application where the retrieval data structure (say, for $v=\Theta(\log n)$) is one subroutine among several other data structures that need to be stored.
We show that if these other data structures occupy at least a comparable size to the retrieval data structure, then retrieval effectively incurs much less redundancy.

More precisely, let $A$ be the array obtained by concatenating the other data structures.
Hence, the query algorithms for these data structures would (only) require random access to $A$.
We show that $A$ and the retrieval data structure can be jointly stored using $|A|w+nv+n^{0.67}$ bits, such that it allows any $A[i]$ to be accessed in constant time, and any retrieval query to be answered in constant time as well.
The extra $n^{0.67}$ bits can be reduced to $\poly\log n$ assuming free random bits.
Comparing this with the lower bound in Theorem~\ref{thm:retrieval_ub_lb}, we demonstrate that merely the existence of one other data structure allows for an improvement on the space usage of retrieval, without downgrading the performance of either data structure by more than a constant factor.
A similar phenomenon, known as \emph{catalytic computation}, has been observed in space-bounded computation~\cite{buhrman2014computing}, although as far as we understand, there is no overlap in the techniques.

Our construction thinks of the random accesses to $A$ as themselves being retrieval queries, where the indices $i$ are the keys and the contents $A[i]$ are the values. What distinguishes this from the classical retrieval problem is that, whereas in the classical problem all keys are arbitrary, in this problem there are $|A|$ keys that are \emph{guaranteed to be contained} in the input. 
It turns out that this simple distinction is enough to bypass our retrieval lower bound, obtaining a data structure with fast queries and low redundancy.

We refer to this problem as the \defn{augmented retrieval problem}, and we refer to the queries to $A$ as \defn{augmented queries}. In this paper, we only consider the static setting, where both the retrieval and the augmented array are immutable.

\begin{restatable}{theorem}{ThmAugmentedRetrieval}
    \label{thm:augmented_retrieval_intro}
    Let parameters $U = n^{1 + O(1)}$ and $V = n^{3 + O(1)}$. 
    For any constant $\eps > 0$, there is a static augmented retrieval data structure 
    in the word RAM model with word size $\Theta(\log n)$ such that:
    \begin{itemize}
    \item The data structure answers value queries for a set of $n-\eps n$ keys, where the keys are from a universe $[U]$ and values are from $[V]$.
    \item The data structure supports $\eps n$ augmented queries, which allow the user to access an array of $\eps n$ elements $a_1, \ldots, a_{\eps n} \in [V]$.
    \item The data structure answers any retrieval and augmented query in worst-case constant time, and uses $n\log V+O(\log^2 n)$ bits of space.
    \item The data structure assumes free access to random hash functions and random permutations. Over this randomness, the data structure succeeds with probability $1-O(n/V)$ in its construction process.
    \end{itemize}
\end{restatable}

We remark that a special case of the augmented retrieval problem was also studied in~\cite{hu2025optimal}, who showed a result analogous to Theorem \ref{thm:augmented_retrieval_intro} in the case where \emph{almost all} of the queries are augmented queries (if there are $m$ non-augmented queries, there must be $U \log m$ augmented queries).\footnote{The authors of \cite{hu2025optimal} used augmented queries as an intermediate result to construct a highly space-efficient static dictionary. Based on this, we suspect that augmented queries have the potential to be quite useful for designing other space-efficient data structures in the future that use retrieval internally.} In contrast, Theorem \ref{thm:augmented_retrieval_intro} says that even a small linear number of augmented queries is enough to eliminate the redundancy for the retrieval queries. We also prove (Theorem~\ref{thm:augmented_lb}) that $\Omega(n)$ augmented queries are needed. More generally, assuming $V = n^{3 + O(1)}$ and a machine word size of $w = \Theta(\log n)$, it turns out that there is a threshold behavior: For any constant-time solution, if there are $o(n)$ augmented queries, then they do not help at all; but once there are $\eps n$ such queries for any positive constant $\eps$, we are suddenly able to reduce the redundancy all the way down to $\log^2 n$.

Finally, it is worth making a few comments on how to interpret Theorem \ref{thm:augmented_retrieval_intro}. For simplicity of notation, the theorem uses $n$ to denote the total number of queries including the augmented queries.
The theorem also assumes that the augmented queries have the same value range as the regular retrieval queries. This is without loss of generality when they both have $\Theta(\log n)$ bits, by a simple change-of-address calculation on the augmented array. Finally, although the theorem assumes free random bits, this assumption can be removed at a cost of increasing the space to $n\log V+O(n^{0.67})$ (see Corollary~\ref{col:augmented_retrieval_remove_random}).
We also handle smaller $V$ in Corollary~\ref{col:augmented_retrieval_small_size}.

\paragraph{Upper and lower bounds for static filters.}
Closely related to the static retrieval problem is the problem of constructing an approximate set membership data structure (a.k.a., a filter) \cite{bloom1970space,dietzfelbinger2008succinct, porat2009optimal,pagh2005optimal, arbitman2010backyard, bender2022optimal,carter1978exact,pagh2004cuckoo, fan2014cuckoo, chen2017dynamic, luo2019consistent}. A filter is initialized with two inputs---a set $S$ of $n$ elements from $[U]$ and a false positive rate $\epsilon$. The filter then needs to support queries where $\texttt{query}(x) = \text{true}$ for each $x \in S$, and where $\Pr[\texttt{query}(x) = \text{false}] \ge 1 - \epsilon$ for each $x \in [U] \setminus S$. For sufficiently large $U$, the information-theoretic space requirement for a filter is $n \log \epsilon^{-1} - o(1)$ bits \cite{carter1978exact}.

Assuming free randomness, one can directly transform any retrieval data structure into a filter data structure with the same time and redundancy guarantees.\footnote{Simply map each key $x \in S$ to a random hash $h(x) \in [\epsilon^{-1}]$ in the retrieval data structure, and, in the filter data structure, declare a key $x$ to be present if and only if the retrieval data structure returns $h(x)$.} Although this transformation is widely used in both theory \cite{porat2009optimal,dietzfelbinger2008succinct,dietzfelbinger2019constanttime} and practice \cite{graf2020xor, dillinger2022fast}, it is not known whether it gives an optimal space-time trade-off for filters. 

Our final result says that this transformation does, in fact, give an optimal filter. We show that filter data structures are subject to the same redundancy/time trade-off curve as are retrieval data structures: 
\begin{theorem}[Later Restated as \cref{thm:non_augmented_lb_filter}]
\label{thm:non_augmented_lb_filter_intro}
    Let $\epsilon, n, U$ satisfy $\epsilon > n^2/U$, and consider a machine word size $w \ge \log U$. Then any static filter that stores $n$ keys in $[U]$ with a false positive rate of $\epsilon$, and that incurs $O(t)$ probes per query in the worst case, must use space 
    \[n\log \epsilon^{-1}-O(1)+\lfloor n\cdot e^{-O(wt/{\log \epsilon^{-1}})}\rfloor.\]
\end{theorem}

The above lower bound applies to data structures that are allowed free randomness. However, due to technical reasons, it only applies to data structures offering worst-case query time guarantees instead of expected time guarantees.
Supposing $\log \epsilon^{-1}$ is an integer, we also obtain almost-matching upper bounds, one with free randomness (this is an almost immediate corollary of our upper bound for retrieval), and one with explicit hash functions (this bound comes with an extra $n^{0.67}$ bits of redundancy). 

\subsection{Related Work on Retrieval and Filters}

The earliest (implicit) solutions to the static retrieval problem use minimal perfect hash functions \cite{hagerup2001efficient, mehlhorn1984data, fredman1984size}. Hagerup and Tholey \cite{hagerup2001efficient} give a construction using $n \log e + o(n)$ bits that achieves query time $O(1)$. This implies a static retrieval data structure with redundancy $n \log e + o(n)$ bits and query time $O(1)$. 

The first results to achieve redundancy $o(n)$ appear to be in concurrent works by Dietzfelbinger and Pagh \cite{dietzfelbinger2008succinct} and Porat \cite{porat2009optimal}. Dietzfelbinger and Pagh \cite{dietzfelbinger2008succinct} give a solution with query time $O(k)$ and space $nv \cdot (1 + e^{-k}) + O(\log \log n)$ bits. Porat~\cite{porat2009optimal} reports a solution with $nv + o(n)$ bits and $O(1)$-time queries. A priori, Porat's solution might seem to contradict Theorem \ref{thm:retrieval_ub_lb}. However, a careful reading of the proofs in \cite{porat2009optimal} suggests that there is an implicit assumption (at least) of $v \le O((\log n) / (\log \log n)^2)$.
Finally, in the special case of $O(1)$-bit retrieval (and assuming access to fully random hash functions), Dietzfelbinger and Walzer \cite{dietzfelbinger2019constanttime} show how to achieve $O(1)$ query time with redundancy $O(\log n)$ bits.

A surprising artifact of our lower bound (Theorem \ref{thm:retrieval_ub_lb}) is that, in general, if one wishes to achieve a redundancy of $n^{1 - \Omega(1)}$ bits, and if the machine word size $w$ is $\Theta(\log n)$, then no matter what the value size $v$ is, there is no better solution than to simply use $v$ instances of Dietzfelbinger and Walzer's \cite{dietzfelbinger2019constanttime} data structure appended together. On the flip side, if one wishes to achieve $O(1)$-time queries, and if $w = \Theta(v) = \Theta(\log n)$, then the original minimal perfect hashing solutions \cite{hagerup2001efficient, mehlhorn1984data, fredman1984size} all achieve asymptotically optimal redundancy. 

In addition to static retrieval, researchers have also studied \emph{dynamic retrieval} \cite{demaine2006dictionariis, mortensen2005dynamic,kuszmaul2025tight}, where key-value pairs are inserted and deleted over time, and \emph{value-dynamic retrieval} \cite{kuszmaul2024space}, where the key set remains static but values can be updated. For value-dynamic retrieval, Kuszmaul and Walzer \cite{kuszmaul2024space} establish an information-theoretic lower bound of $nv + \Omega(n)$ bits. In general, Kuszmaul and Walzer's result \cite{kuszmaul2024space} can be viewed as a separation result, separating the space complexity of static retrieval from that of value-dynamic retrieval. However, the results in this paper offer a surprising nuance to this separation: when $v = \Theta(w)$ (where $w$ is the machine-word size), the optimal asymptotic redundancy for any \emph{constant-time} solution is actually the same for both problems. In this parameter regime, minimal perfect hashing is optimal. 

Filters, which were first introduced by Bloom in 1970 \cite{bloom1970space}, have also been the subject of a great deal of work for both theory \cite{dietzfelbinger2008succinct, porat2009optimal,pagh2005optimal, arbitman2010backyard, bender2022optimal,carter1978exact,pagh2004cuckoo, fan2014cuckoo, chen2017dynamic, luo2019consistent,lovett2013space,kuszmaul2024space} and practice \cite{abdennebi2021bloom, luo2019optimizing,clerry1984compact, pagh2005optimal, bender2012dont, pandey2021vector, pandey2017generalpurpose, geil2018quotient,graf2020xor, dillinger2022fast,dietzfelbinger2019constanttime}. In the static setting, the best upper bounds in both theory \cite{porat2009optimal,dietzfelbinger2008succinct,dietzfelbinger2019constanttime} and practice \cite{graf2020xor, dillinger2022fast} are based on direct reductions to the static retrieval problem. In the dynamic setting, where keys can be inserted and deleted, the best known solutions are all based on space-efficient dictionaries that store fingerprints of keys \cite{carter1978exact, pagh2005optimal, arbitman2010backyard, bender2022optimal}; such solutions can achieve redundancy $O(n)$ bits, which is known to be optimal in the dynamic setting \cite{lovett2013space,kuszmaul2024space}. In the incremental setting \cite{bloom1970space, kuszmaul2024space, lovett2013space}, where insertions but not deletions are allowed, the optimal redundancy drops to $o(n)$ bits \cite{kuszmaul2024space} so long as $\epsilon = o(1)$. 

One consequence of our lower bound for filters is that, in the special case where $\epsilon = 1 / \poly n$ (but where $U \ge \epsilon^{-1} n^2$), any static filter that wishes to achieve $O(1)$-time operations must incur $\Omega(n)$ bits of redundancy. In the same parameter regime, it is known how to construct \emph{dynamic} filters \cite{bender2022optimal} with $O(1)$-time operations and $O(n \log \log \cdots \log n)$ bits of redundancy (for any constant number of logarithms). Thus, there is a surprising sense in which, for small values of $\epsilon$, the optimal time-efficient solutions to the static and dynamic filter problems have almost the same redundancy bounds as one another.

\section{Lower Bounds}

\subsection{Lower Bound for Non-Augmented Retrieval}

We first show a space lower bound for the static retrieval problem. In particular, when the value size is $\Theta(w)$ bits, where $w$ is the machine-word size, and queries are answered in constant time, we show that retrieval data structures incur $\Omega(n)$ bits of redundancy.

\begin{theorem}
    \label{thm:non_augmented_lb}
    Let $U$ be the universe size such that $2n \le U \le n^{O(1)}$, and let $V \le n^{O(1)}$ be the value range. Let $v = \log V$ for notational simplicity. Suppose that there exists a static retrieval data structure in the cell-probe model with word size $w \ge v$ that stores $n$ key-value pairs, where the keys are from $[U]$ and the values from $[V]$. The data structure assumes free access to an infinitely long random tape, which is its only source of randomness, and answers queries using $t$ probes in expectation.
    Then, this retrieval data structure must use at least $nv + \lfloor n \cdot e^{-O(wt/v)} \rfloor$ bits of space.
\end{theorem}

By saying that the data structure answers queries ``using $t$ probes in expectation,'' we mean that for a fixed set of inputs and a fixed query, the expected number of probes required to answer the query is at most $t$, where the expectation is taken over the randomness of the random tape.

\begin{proof}
    Let $\beta$ be a sufficiently large constant. We aim to prove a space lower bound of $nv + \lfloor n \cdot e^{-\beta wt/v} \rfloor$ bits. We only need to consider the case where $wt/v \in [1, (1/\beta) \ln n]$, since otherwise we have $e^{-\beta wt/v} < e^{-\ln n} = 1/n$, and our desired space bound becomes $nv$ bits, which is simply the information-theoretic lower bound for retrieval data structures.

    \paragraph{Constructing a communication game.}
    Our lower bound is proven by analyzing a one-way communication game. Let $X = \{x_1,\dots,x_n\} \subset [U]$ be a set of keys, and let $A = (a_1,\dots,a_n) \in [V]^n$ be their associated values. Alice wishes to send $X$ and $A$ to Bob. We also assume that Alice and Bob both have access to an infinitely long public random tape. By information-theoretic lower bounds, in order to send $X$ and $A$, Alice needs to send at least $\log \binom{U}{n} + nv$ bits in expectation, where the expectation is taken over the randomness of the public random tape.

    We prove \cref{thm:non_augmented_lb} by contradiction. Assume there exists a retrieval data structure that answers queries in $t$ probes in expectation and uses $nv + R$ bits of space, where $R \le n \cdot e^{-\beta \cdot wt/v}$. We will show that this assumption allows us to construct a communication protocol where Alice can send the key set $X$ and associated values $A$ to Bob using a message whose expected length is less than $\log \binom{U}{n} + nv$ bits,\footnote{Additionally, Alice's message will be a prefix-free code, meaning that Bob can determine where the end of the message is without additional bits of information being sent to encode the length.} which leads to a contradiction.

    \paragraph{Notation.} We start by setting up notations that we will use in the proof. Let $D$ be the memory configuration of the retrieval data structure, when the input is $X$ and $A$, and the construction algorithm uses the public random tape as its source of randomness. By Markov's inequality, we can show that, with probability $\ge 4/5$ over the public random tape, the data structure satisfies the following:
    \begin{property}
        \label{prop:avg_time_bound}
        The following holds for the data structure:
        \begin{enumerate}
            \item The average query time of the queries in $X$ is at most $10t$;
            \item The average query time of the queries in $[U]$ is at most $10t$.
        \end{enumerate}
    \end{property}
    In the protocol, if \cref{prop:avg_time_bound} does not hold, Alice will send one bit indicating that this is the case, then send the entire input $X$ and $A$, and terminate. That is, we only try to save communication bits when \cref{prop:avg_time_bound} holds.
    
    Throughout our proof, we will fix the retrieval data structure $D$ and perform various queries on it (using the public random tape as its source of randomness), sometimes omitting the notation $D$ for the data structure. We also define the following notations and terminologies describing the behavior of the queries:
    \begin{itemize}
    \item Let $\tslow\defeq t\cdot e^{(\beta/2) wt/v}$ be a threshold. If a query probes at least $\tslow$ cells, we say that it is a \defn{slow query}, otherwise we say that it is a \defn{fast query}.
    \item We define the \defn{truncated queries}, where the truncated version of a query on key $x$ simulates the original query on $x$ until it probes $\tslow$ cells. If a query probes fewer than $\tslow$ cells, the truncated version returns the correct answer; otherwise, it returns $\bot$. In other words, the truncated version leaves fast queries unchanged while truncating slow queries. We study truncated queries because they provide worst-case time guarantees, which will make it possible to apply certain concentration bounds in the analysis. All subsequent definitions are based on truncated queries.
    \item For any $S\subset [U]$, let $\pb(S)$ be the set of cells that are probed when answering the \emph{truncated} queries in $S$ on $D$, and let $f(S)=|\pb(S)|$ be the number of cells in $\pb(S)$.
    \item Define the \defn{footprint} of $S$, denoted by $\ft(S)$, to be a sequence of $f(S)$ cells, constructed as follows: Initially, $\ft(S)$ is empty. Enumerate the keys $x\in S$ in increasing order. For each $x$, simulate its \emph{truncated} query on $D$ one probe at a time. If the cell it probes is already added to $\ft(S)$, ignore it. Otherwise, append its content (but not the address) to $\ft(S)$. We remark that $\ft(S)$ has length exactly $f(S)$ machine words; and that, given $S$, and given $\ft(S)$, one can straightforwardly recover the answer to every truncated query for every element in $S$. This also implies that $\ft(S)$ is a prefix-free code. That is, given $S$ and $\ft(S)$, one can determine where the end of $\ft(S)$ is without additional information.
    \item Let $M\defeq (nv+R)/w$ be the number of cells in the data structure.
    \end{itemize}
    
    \paragraph{Motivating the communication protocol.} Our proof will take the form of a communication protocol in which Alice exploits the space efficiency of $D$ to send the keys $X$ and values $A$ to Bob with impossibly few bits. To motivate the protocol, let us imagine for a moment that all queries are fast queries, and that we are in the parameter regime $w = v = \Theta(\log n)$ and $t = O(1)$.  
    
    Suppose that the number $f(X)$ of distinct memory cells probed by all queries satisfies $f(X) < nv/w = n$. Then, after Alice sends Bob the key set $X$, there is an opportunity to compress the answers to the queries on $X$ as follows: Alice can simply send $\ft(X)$ to Bob, which consists of $f(X) < nv/w$ machine words, but allows Bob to recover answers to all of the queries in $X$. 
    
    On the other hand, if $f(X) \ge n$, then we can argue that the retrieval data structure $D$ implicitly encodes $\Omega(n)$ bits of information about $X$. This is because, if we take a random set $S \subseteq [U]$ of size $n$, the probability of $f(S)$ satisfying $f(S) \ge n$ is very small (all but a $e^{-\Omega(n)}$ fraction of sets satisfy $f(S) \le (1 - \Omega(1))n$). So the fact that $f(X) \ge n$ is enough for us to recover $\Omega(n)$ bits of information about $X$ from $D$. This allows Alice to construct a message using $D$ in which she is able to save $\Omega(n)$ bits of space over the information-theoretic optimum. 

    In the rest of the section, we will present the full protocol, and we will show how to remove all of the above assumptions (allowing for slow queries and supporting all parameter regimes).

    \paragraph{Defining the full protocol.} As discussed earlier, the first step of the protocol is to send one bit indicating whether \cref{prop:avg_time_bound} holds. If not, then Alice directly sends $X$ and $A$, and terminates.
    
    Let
    \[
    T\defeq(nv-(10t/\tslow)nv-100)/w
    \numberthis \label{eq:defn_T}
    \]
    be a threshold. Intuitively, $T$ represents the number of cells needed to encode the answers to the fast queries.
    Next, Alice sends one bit to Bob, indicating whether $f(X)\ge T$. What Alice sends afterwards depends on this bit. 

    \paragraph{Case 1: $f(X)<T$.} In this case, we can send the values $A$ more efficiently, by sending the contents of the cells probed by $X$. That is, Alice's message to Bob consists of the following:
    \begin{enumerate}
        \item One bit indicating that \cref{prop:avg_time_bound} holds;
        \item One bit indicating that $f(X) < T$;
        \item The key set $X$, using $\ceil*{\log\binom{U}{n}}$ bits;
        \item $\ft(X)$, using $f(X)\cdot w$ bits;
        \item The answers to the slow queries in $X$, in the increasing order of their keys. \cref{prop:avg_time_bound} implies that there are at most $(10t/\tslow)n$ slow queries, so this uses at most $\lceil(10t/\tslow)nv\rceil$ bits.
    \end{enumerate}
    
    After receiving $X$ and $\ft(X)$, Bob can determine which queries in $X$ are fast or slow, and can recover the answers to the fast queries. After that, Bob receives the answers to the slow queries, at which point he has fully recovered both $X$ and $A$. Overall, the size of the message is bounded by
    \begin{align*}
        & 2+\ceil*{\log\binom{U}{n}}+f(X)\cdot w+\ceil*{(10t/\tslow)nv} \\
        <{}& \log\binom{U}{n}+T\cdot w+\bk*{4+(10t/\tslow)nv}\\
        = {}& \log\binom{U}{n}+(nv-(10t/\tslow)nv-100)+\bk*{4+(10t/\tslow)nv}\tag{by \eqref{eq:defn_T}} \\
        < {}& \log\binom{U}{n}+nv-10
    \end{align*}
    bits.

    \paragraph{Case 2: $f(X)\ge T$.} We show that, conditioned on $D$, the number of possible key sets $X$ with the property that $f(X)\ge T$ is small, so that conditioned on the event that $f(X)\ge T$, we can send the set $X$ more efficiently.
    
    Formally, let $S\subseteq [U]$ be a random key set, where each key is included independently with probability $n/U$. Given that $f(X)\ge T$, Alice can send $X$ using $-\log\Pr\Bk[\big]{S=X \,\big|\, f(S)\ge T}+O(1)$ bits. To prove that this encoding is efficient, we will show that $\Pr[f(S)\ge T]$ is small, and that $\Pr[S=X]$ is close to $1/\binom{U}{n}$. Therefore, we will have that $-\log\Pr\Bk[\big]{S=X \,\big|\, f(S)\ge T} = -\log \Pr[S=X]+\log \Pr[f(S)\ge T]$ is much smaller than $\log\binom{U}{n}$.
    
    To prove that $\Pr[f(S)\ge T]$ is small, we show that $f$ is tightly concentrated around its mean, and that $\E[f(S)]$ is much smaller than $T$. From this, we will be able to deduce that the probability of $f(S)\ge T$ is exponentially small.

    To show that $f$ has a concentration property, we first note that $f$ is a monotone submodular function. This is because $f$ is a coverage function, i.e., $f(S)$ is the size of the union of the sets $\pb(x)$ for each $x$ in $S$. Such functions are shown to be monotone submodular in \cite{krause2014submodular}. Also, $f$ has marginal values in $[0,\tslow]$, which is to say that for any set $S$ and any element $x$, the value of $f(S\cup\{x\})$ is at most $f(S)+\tslow$.

    Those properties of $f$ are sufficient for deriving a concentration inequality, as shown in the following lemma.
    \begin{lemma}[\cite{vondrak2010note}, rephrased]
        \label{lem:conc_submodular}
        If $Z=f(X_1,\dots,X_n)$ where $X_i\in \{0,1\}$ are independently random and $f$ is a monotone submodular function with marginal values in $[0,t]$, then for any $\delta>0$,
        \[
        \Pr[Z\ge (1+\delta)\E[Z]]\le e^{-\frac{\delta^2}{\delta+2}\frac 1t \E[Z]}
        \]
    \end{lemma}

    Now, we apply \cref{lem:conc_submodular} to show that $\Pr[f(S)\ge T]$ is small. We first show that $\E[f(S)]$ is small:

    \begin{claim}
        \label{clm:efs_small}
        We have that $\E[f(S)]\le M(1-e^{-20wt/v})$.
    \end{claim}
    \begin{proof}
        For each cell $i\in [M]$, let $g_i=\sum_{x\in [U]}\ind[i\in \pb(x)]$ be the number of truncated queries that probe $i$. Then we have $\sum_i g_i\le 10Ut$ by \cref{prop:avg_time_bound}, which states that the average query time is at most $10t$. Note that $\E[f(S)]$ is just the sum of the probability that each cell is probed by a query in $S$, and we have that
        \begin{align*}
            \E[f(S)]&=\sum_{i=1}^{M} \bk*{1-\bk*{1-\frac nU}^{g_i}} \\
            &\le M\bk*{1-\bk*{1-\frac nU}^{10Ut/M}} \\
            &\le M(1-e^{-20nt/M})\tag{$n/U\le 1/2$} \\
            &\le M(1-e^{-20wt/v}),\tag{$M\ge nv/w$}
        \end{align*}
        where the second line follows from Jensen's inequality, which is applicable because $1 - (1-n/U)^{g_i}$ is concave in $g_i$.
    \end{proof}

    Secondly, we show that $T$ is large relative to $\E[f(S)]$:
    \begin{align*}
        T&=(nv-(10t/\tslow)nv-1)/w \\
        &=M-(R+(10t/\tslow)nv+1)/w \tag{$M=(nv+R)/w$}\\
        &\ge M-(nv/w)\cdot e^{-(\beta/3)wt/v}\tag{$R\le n\cdot e^{-\beta wt/v}$ and $\tslow=t\cdot e^{(\beta/2)wt/v}$}  \\
        &\ge M(1-e^{-(\beta/3)wt/v}).\tag{$M\ge nv/w$}
    \end{align*}
    
    Let $\delta$ be a parameter such that $T=(1+\delta)\E[f(S)]$, then $\delta\ge e^{-20wt/v}$. We have by \Cref{lem:conc_submodular} that 
    \begin{align*}
        \Pr[f(S)\ge T]&\le \exp\bk*{-\frac{\delta^2}{\delta+2}\frac 1{\tslow}\E[f(S)]} \\
        &\le \exp\bk*{-\frac 16(nv/(wt))\cdot e^{-(40+\beta/2)wt/v}}\tag{$\E[f(S)]\le M\le 2nv/w$}\\
        &\le \exp\bk*{-n\cdot e^{-(2\beta/3)wt/v}} \tag{$wt/v\ge 1$}.
    \end{align*}
    Now since $X$ satisfies that $f(X)\ge T$, we can send $X$ using a message of length
    \begin{align*}
    -\log \Pr \Bk[\big]{S=X \,\big|\, f(S)\ge T} + O(1) &= -\log \Pr[S=X]+\log \Pr[f(S)\ge T]+O(1)
    \end{align*}
    bits. Here,
    \begin{align*}
        -\log \Pr[S=X]&=-\log\bk*{\bk*{\frac nU}^n\bk*{1-\frac nU}^{U-n}}.
    \end{align*}
    In comparison, using Stirling's approximation, we can show that
    \begin{align*}
        \log \binom{U}{n}&=-\log\bk*{\bk*{\frac nU}^n\bk*{1-\frac nU}^{U-n}}-O(\log n).
    \end{align*}
    Therefore $-\log \Pr[S=X]=\log \binom{U}{n}+O(\log n)$. 
    
    So the number of bits saved when sending $X$ is 
    \begin{align*}
        &\log\binom{U}{n}-\bk[\big]{-\log \Pr[S=X]+\log \Pr[f(S)\ge T]+O(1)} \\
        \ge & -\log\Pr[f(S)\ge T]-O(\log n) 
        =  \log e\cdot n\cdot e^{-(2\beta/3)wt/v}-O(\log n),
    \end{align*}
    which is larger than $n\cdot e^{-(3\beta/4)wt/v}$ since we assumed that $wt/v\le (1/\beta)\ln n$.

    Overall, in the case that $f(X)\ge T$, Alice's message to Bob consists of the following:
    \begin{enumerate}
        \item One bit indicating that \cref{prop:avg_time_bound} holds;
        \item One bit indicating that $f(X)\ge T$;
        \item The memory configuration $D$, using $nv+R$ bits;
        \item The set $X$ conditioned on $D$ and the event that $f(X)\ge T$, using $\le \log\binom{U}{n}-n\cdot e^{-(3\beta/4)wt/v}$ bits.
    \end{enumerate}
    After receiving $D$ and $X$, Bob can simulate the \emph{untruncated} queries and learn the values $A$, at which point he has recovered both $X$ and $A$, as desired. Since $R\le n\cdot e^{-\beta \cdot wt/v}$ and $\beta w t / v \le \ln n$, this protocol uses fewer than $\log\binom{U}{n}+nv-10$ bits.

   \paragraph{Putting the pieces together.}
   Finally, we calculate the expected length of the message sent by Alice.

   In the case where \cref{prop:avg_time_bound} does not hold, Alice sends $1+\ceil*{\log\binom{U}{n}+nv}$ bits in total. This happens with probability $\le 1/5$.
   In the two cases where \cref{prop:avg_time_bound} holds, Alice sends fewer than $\log\binom{U}{n}+nv-10$ bits. Overall, the expected number of bits sent is at most
   \begin{align*}
       \log\binom{U}{n}+nv+\bk*{2\cdot \frac 15-10\cdot \frac 45}<\log\binom{U}{n} + nv.
   \end{align*}
   Therefore, we obtain an impossibly good communication protocol when $R$ is at most $n\cdot e^{-\beta\cdot wt/v}$.
\end{proof}

\subsection{Two Extensions: Lower Bounds for Augmented Retrieval and Filters}

The lower bound technique used in the proof of \cref{thm:non_augmented_lb} can be extended to get the following two results. The proofs of those lower bounds are deferred to \cref{app:lb}.

\paragraph{A lower bound for augmented retrieval.} Later in this paper, we will present an augmented retrieval data structure with almost no redundancy, for the parameter regime where the $n$ retrieval values are $\Theta(\log n)$ bits and where there are $\Omega(n)$ augmented queries, each of which also accesses a $\Theta(\log n)$-bit value.

A natural question is whether $\Omega(n)$ augmented queries are really necessary in order to achieve almost zero redundancy. Here, we show that $o(n)$ augmented queries do not suffice to get $o(n)$ redundancy. 
\begin{restatable}{theorem}{AugmentedLB}
    \label{thm:augmented_lb}
    Let $U$ be the universe size such that $2n \le U \le n^{O(1)}$, and let $V\le n^{O(1)}$ be the value range. Let $v=\log V$ for simplicity.
    In the cell-probe model with word size $w\ge v$, suppose that there exists a static augmented retrieval data structure that stores $n$ key-value pairs---where the keys are from $[U]$ and the values are from $[V]$---as well as $m$ augmented elements in $[V]$. The data structure assumes free access to an infinitely long random tape, which is its only source of randomness, and answers queries using $t$ probes in expectation. If $m\le n\cdot e^{-O(wt/v)}$, then this augmented retrieval must use $(n+m)v+\lfloor n\cdot e^{-O(wt/v)}\rfloor$ bits of space.
\end{restatable}

We remark that this lower bound holds even when there is no time constraint on the augmented queries.

\paragraph{A lower bound for filters.} We show that the same redundancy/time trade-off curve that we proved for retrieval data structures also holds for filters, where $\log \epsilon^{-1}$ takes the role of $v$.

\begin{restatable}{theorem}{FilterLB}
    \label{thm:non_augmented_lb_filter}
    Let $n^2 \le U\le n^{O(1)}$ be the universe size. In the cell-probe model with word size $w \ge \log U$, suppose that there exists a static filter that stores $n$ keys in $[U]$ and has a false positive rate of 
    $\epsilon \ge n^2 / U$.
    The data structure assumes free access to an infinitely long random tape, which is its only source of randomness, and answers queries using $t$ probes in expectation. Then this filter must use
    \[
    n\log \epsilon^{-1} - O(1) + \floor*{n \cdot e^{-O(wt/{\log \epsilon^{-1}})}}
    \] bits of space.
\end{restatable}

\section{Upper Bounds for Augmented Retrieval}
\label{sec:augmented_ub}

In this section, we prove \cref{thm:augmented_retrieval_intro}.

\ThmAugmentedRetrieval*

Later on, we will also present two extensions of Theorem \ref{thm:augmented_retrieval_intro}, one that uses explicit random bits (Corollary \ref{col:augmented_retrieval_remove_random}) rather than free randomness, and one that allows for small value-universe sizes $V$ (Corollary \ref{col:augmented_retrieval_small_size}).

The main technical part of the proof is to establish a special case of \cref{thm:augmented_retrieval_intro} where we assume that $V$ is a prime power and that the memory words are elements from a field $\F$ of order $V$.  This special case is formally stated as follows.

\begin{theorem}%
    \label{thm:augmented_retrieval}
    Let $U = n^{1 + O(1)}$, and let $V = n^\gamma$ be a prime power, where $\gamma \ge 3$ is a constant parameter. Let $\F$ be a finite field of order $V$. Let $c>1$ be a constant parameter. There is a static augmented retrieval data structure such that
    \begin{itemize}
    \item The data structure answers value queries for a set of $n-n/c$ keys, where the keys are from a universe $[U]$ and values are from $\F$.
    \item The data structure supports $n/c$ augmented queries, which allow the user to access an array of $n/c$ elements $a_1, \ldots, a_{n/c} \in \F$.
    \item The data structure uses $n$ memory words in $\F$ (i.e., it incurs \emph{no} redundancy), and can answer any retrieval and augmented query in time $\poly(c,\gamma)$. 
    \item The data structure assumes free access to random hash functions and random permutations. Over the randomness of these hash functions, the data structure succeeds with probability $1-O(n/V)$ in its construction process.
    \end{itemize}
\end{theorem}

We first assume \cref{thm:augmented_retrieval} and show how it implies \cref{thm:augmented_retrieval_intro}.

\begin{proof}[Proof of \cref{thm:augmented_retrieval_intro}]
    Given any $V=n^{3+O(1)}$, we first round $V$ up to the nearest prime $V'$, and use the construction in \cref{thm:augmented_retrieval} while assuming that all values are in $[V']$ instead of $[V]$. The following lemma upper bounds $V'$.

    \begin{lemma}[\cite{qi1992chebychevs}]\label{lem:nextprime}
        For sufficiently large $V$, there exists a prime in $[V, \, V + V^{7/11}]$.
    \end{lemma}

    Then, we apply \cref{thm:augmented_retrieval} with $c, \gamma = O(1)$, which gives us a retrieval data structure whose memory cells store field elements in $\F$, such that each query takes constant time to answer.

    Finally, we need to convert the memory model in \cref{thm:augmented_retrieval} into the word RAM model where each memory cell stores an integer in $[2^w]$. For this, we use Theorem 1 of \cite{dodis2010changing}, which states that it is possible to store a sequence of $n$ elements in $[V']$ in the word RAM model with word size $w=\Theta(\log n)$, while supporting read and write operations in constant time, using $n\log V'+O(\log^2 n)$ bits of space. Since $V'\le V+V^{7/11}$ and $V\ge n^3$, the total space used is
    \begin{align*}
        &n\log V' + O(\log^2 n) \\
        ={}&n\log V+n\cdot O(V^{-4/11})+O(\log^2 n) \\
        ={}&n\log V+O(\log^2 n)
    \end{align*}
    bits.
\end{proof}

\subsection{Algorithm Overview}

Now we prove \cref{thm:augmented_retrieval}.
Our algorithm views the memory as a vector $\mathbf v\in \F^n$. Each query is associated with a sparse vector that is uniquely determined by the queried key and the hash functions, and can be computed efficiently. A query is answered by computing the inner product of the memory vector and the sparse vector corresponding to the query. 

For our construction to work, the set of sparse vectors corresponding to the $n$ queries (including both the retrieval and augmented queries) must be linearly independent. This condition ensures that, no matter what value is associated with each query, we can configure the data structure so that each query returns the correct value.

We will construct a matrix with $U+n/c$ rows and $n$ columns, where each row is a sparse vector corresponding to a query. We call the $n/c$ rows corresponding to the augmented queries \defn{augmented rows}, and the $U$ rows corresponding to retrieval queries \defn{retrieval rows}.

The matrix should satisfy the following property: For any fixed set of $n-n/c$ retrieval rows, the $n\times n$ submatrix consisting of both those rows and the augmented rows has full rank with probability $1-O(n/V)$. If this property holds, then for any set of valid queries (including $n - n/c$ retrieval queries and all $n/c$ augmented queries), with probability $1-O(n/V)$, we can find a vector $\mathbf v$ that correctly answers all the queries, which implies \cref{thm:augmented_retrieval}. We now design an algorithm to generate the matrix.

\paragraph{Constructing the matrix.}

Let $\alpha$ be a parameter that will later be set to a sufficiently large constant (but that, for now, will be included in asymptotic notation). For each row, we designate certain entries as \defn{active entries} and we fill in each active entry with a uniformly random element from $\F$ (the finite field of size $V$). All remaining entries are set to $0$. The active entries are selected as follows:

\begin{itemize}
\item Let $\tcoup=\alpha\cdot \max\{c,\gamma\}$. For each row (both augmented and retrieval rows), we randomly select $\tcoup$ entries with replacement and activate them. We call these \defn{coupon entries}.
\item For augmented rows, we designate additional active entries. Let $\tperm=\alpha\cdot c^2\cdot\gamma$. We partition the columns into $n/c$ blocks of $c$ columns each. We then generate $\tperm$ permutations $\{p_i\}_{i\le \tperm}$ over $[n/c]$, where each permutation bijectively maps the $n/c$ augmented queries to the $n/c$ blocks. For each augmented row with index $j \in [n/c]$, we activate all entries in the blocks $p_i(j)$ $(1 \le i \le \tperm)$. We call these \defn{permutation entries}.
\end{itemize}
If an entry is activated in both steps, it is considered both a coupon entry and a permutation entry.

\paragraph{Why augmented queries help.} Before continuing, it is worth taking a moment to briefly motivate the role of augmented queries in the construction, and, in particular, how the augmented queries let us avoid the lower bound (Theorem \ref{thm:non_augmented_lb}) that holds for the (non-augmented) retrieval problem.

In the proof of Theorem \ref{thm:non_augmented_lb}, a key insight was the following. Consider any retrieval data structure $D$ with, say, $O(1)$-time queries; and consider a random set $S \subseteq [U]$ of $n$ keys. The probability that the queries to the keys $S$ manage to probe every cell in $D$ (or even a $(1 - o(1))$ fraction of the cells) is very small. Therefore, \emph{almost all sets} $S$ of keys are incapable of using the data structure $D$ effectively. This allowed us in the proof of Theorem \ref{thm:non_augmented_lb} to obtain a lower bound by either compressing $D$ (if $S$ fails to probe more than $(1 - o(1))$ fraction of the cells) or by compressing $S$ (if $S$ is one of the very few sets that manages to probe a $(1 - o(1))$ fraction of the cells in $D$). 

Now, let us consider how this situation changes if we replace $n/c$ of the queries in $S$ with augmented queries, letting a user query a number $i \in [n/c]$ to recover a value $a_i$. Because the augmented queries are on a fixed set of keys $[n/c]$, it is easy to implement them so that, collectively, they probe every cell in $D$. (Intuitively, this is the reason that augmented queries include permutation entries in addition to coupon entries.) This allows us to design a data structure in which the retrieval queries extract useful information from the cells that they probe (even though those cells represent a $1 - \Omega(1)$ fraction of the cells in $D$), while the augmented queries extract useful information from the remaining cells.

Perhaps surprisingly, the inclusion of the permutation entries in the above construction is all we need to implement this idea. With these entries in place, we can show that the resulting $n \times n$ matrix is (with good probability) full rank, and, in doing so, bypass the lower bound that holds for retrieval queries. 

\subsection{Analysis of the Matrix}

Now we prove that the construction above has the desired property.

\begin{lemma}
\label{lem:perm_construction_full_rank}
The construction satisfies the following property:
For any fixed set of $n-n/c$ retrieval queries, the $n \times n$ matrix formed by the rows corresponding both to these retrieval queries and to all the augmented queries has full rank with probability at least $1-O(n/V)$.
\end{lemma}

We prove \cref{lem:perm_construction_full_rank} by showing that the matrix formed by both the $n-n/c$ retrieval rows and the $n/c$ augmented rows has a non-zero determinant with probability $1-O(n/V)$. This is accomplished by viewing the determinant as a polynomial. When generating the matrix, if the algorithm places a formal variable $x_{i,j}$ at each active entry $(i,j)$ instead of a random element from $\F$, then the determinant becomes a polynomial of degree $n$. The actual determinant is an evaluation of this polynomial when all variables are assigned random values. According to the Schwartz-Zippel lemma (e.g., \cite[Theorem 7.2]{motwani1995randomized}), a non-zero multivariate polynomial $P$ of degree $n$ over the field $\F$ evaluates to zero with probability at most $n/V$ under random assignments. Therefore, to prove \cref{lem:perm_construction_full_rank}, it suffices to show that the formal polynomial is non-zero with high probability.

\begin{claim}\label{claim:det_non_zero_poly}
Suppose that in the algorithm above, a formal variable $x_{i,j}$ is placed at each active entry instead of a random value from $\F$. Then, for any set of $n-n/c$ retrieval queries, the determinant of the matrix formed by these retrieval rows and all the augmented rows is a non-zero polynomial with probability $1-O(1/V)$.
\end{claim}
Proving \cref{claim:det_non_zero_poly} is equivalent to showing that there exists some degree-$n$ monomial of the form $\prod_{i = 1}^n x_{i,\sigma_i}$, where $(\sigma_1, \ldots, \sigma_n)$ is a permutation, and such that the monomial has a non-zero coefficient (i.e., all of the entries $(i, \sigma_i)$ are active). Since each such a monomial corresponds to a matching between rows and columns, using Hall's theorem, we can reduce \cref{claim:det_non_zero_poly} to the following lemma.

\begin{lemma}\label{lem:perm_construction_span}
Let $S$ be a subset of rows, and let $C_S$ denote the random variable counting the number of columns that contain at least one active entry from rows in $S$.
For any fixed set of $n-n/c$ retrieval rows, the matrix formed by both these retrieval rows and all the augmented rows satisfies the following with probability $1-O(1/V)$: for any subset $S$ of rows in this matrix, $C_S \geq |S|$.
\end{lemma}

In the remaining part of this subsection, we fix a set of retrieval rows and prove \cref{lem:perm_construction_span} by a union bound over the failure probability of all row subsets $S$.

\paragraph{Overview.}
The two types of active entries are designed to span the columns for different sizes of row subsets. Coupon entries are present in every row, and by themselves are sufficient to cover enough columns for any small set; however, as per the coupon collection problem, they cannot cover enough columns when $|S|\geq n-o(n)$. The permutation entries occur only in augmented rows and can cover all columns due to the properties of permutations. For cases where $|S|$ is very close to $n$, we will show that the permutation entries in $S$ can cover many columns.
In the following, we bound the failure probability by discussing two cases: the small case ($|S|\leq n-n/(2c)$) and the large case ($|S| > n-n/(2c)$).

\paragraph{Small case: Coupon entries.} We first bound the failure probability for all subsets containing $s\leq n-n/(2c)$ rows. For such a subset, these rows jointly sample $s\cdot \tcoup$ random coupon entries. Intuitively, these coupon entries cover $s$ columns with high probability since $s$ is not too close to $n$.

In the small case, a violation of the constraint $C_S \geq |S|$ is witnessed by a subset $S$ of $1\le s\le n-n/(2c)$ rows and a subset $T$ of $s$ columns, such that all coupon entries sampled by these $s$ rows fall within these $s$ columns. We apply a union bound over all possible pairs $(S, T)$ to show that with high probability, all row subsets succeed. For a given pair of sets $(S,T)$, the probability of it becoming a witness is $(s/n)^{\tcoup\cdot s}$. Therefore, the probability that any small set fails is bounded by
\begin{align*}
&\sum_{s=1}^{n-n/(2c)} \binom ns^2 (s/n)^{\tcoup\cdot s}\\
\leq{}&\sum_{s=1}^{n-n/(2c)} (en/s)^{2s} (s/n)^{\tcoup\cdot s}\tag{$\binom ns\leq (en/s)^s$}\\
\leq{}&\sum_{s=1}^{n-n/(2c)} (n/s)^{(16c+2)s} (s/n)^{\tcoup\cdot s}\tag{$n/s\geq 1+1/(2c)$}\\
\leq{}&\sum_{s=1}^{n-n/(2c)}(s/n)^{\Omega(\alpha\cdot \gamma\cdot s)}\tag{$\tcoup=\alpha\cdot\max\{c,\gamma\}$} \\
\leq{}&n\cdot\max\BK*{(1/n)^{\Omega(\alpha\cdot \gamma)}, (1-1/(2c))^{\Omega(\alpha\cdot \gamma\cdot n)}}\\
\leq{}&1/n^{\Omega(\alpha\cdot \gamma)}=O(1/V),
\end{align*}
where the third line holds because $n/s \ge 1 + 1/(2c)$ implies $(n/s)^{16c+2} \ge (n/s)^2 \cdot (1 + 1/(2c))^{16c} \ge (n/s)^2 \cdot e^2$; 
the second-to-last line holds because $\log \bk*{(s/n)^s}=s\log s-s\log n$ is convex in $s$, so the function $(s/n)^s$ reaches its maximum at either $s=1$ or $s=n-n/(2c)$; the last line holds as we set $\alpha$ to a sufficiently large constant.

\paragraph{Large case: Permutation entries.} When $s > n - n/(2c)$, at least $s - (n - n/c) > n / (2c)$ augmented rows are in the set. The following claim shows that, by the permutation construction, the permutation entries of these augmented rows cover at least $s$ columns with high probability.

\begin{claim}\label{clm:perm_cover_large}
With probability $1-O(1/V)$, for every subset $S$ of at least $n/(2c)$ augmented rows, $C_S\ge n-n/c+|S|$.
\end{claim}

\cref{clm:perm_cover_large} implies that, with probability $1-O(1/V)$, for every row subset $S$ of size $> n-n/(2c)$, we have $C_S\ge |S|$. This bounds the failure probability of the large row subsets, which, together with the small case, concludes the proof of \cref{lem:perm_construction_span}. Thus, it only remains to prove \cref{clm:perm_cover_large}.

\begin{proof}[Proof of \cref{clm:perm_cover_large}]
We apply a union bound over the failure probability of all large subsets of augmented rows. Let $t$ be the number of \emph{missing} augmented rows, that is, there are $n/c-t$ augmented rows in the set. We only prove for cases where $t \geq 1$, since the case where $t = 0$ trivially holds.

Fix a subset of $n/c-t$ augmented rows where $t \le n/(2c)$. For each $i \in [n/c]$, let $X_i$ be a random variable denoting whether the $i$-th column block is \emph{not} covered by these rows. These $n/c-t$ rows cover fewer than $n-t$ columns when $\sum_{i=1}^{n/c} X_i > t/c$.

It turns out that the variables $X_i$ are \emph{negatively associated (NA)}, which allows us to apply concentration inequalities. The detailed proof of the following claim is deferred to \cref{app:NA_concentration}.

\begin{restatable}{claim}{NAConcentration}
    \label{clm:NA_concentration}
    Let $\mu=\E[\sum_{i=1}^{n/c}X_i]$, then for any $k\ge 2$,
    \begin{align*}
        \Pr\Bk*{\sum_{i=1}^{n/c}X_i>k\mu}\le \bk*{\frac 1k}^{\Omega(k\mu)}.
    \end{align*}
\end{restatable}

Since each permutation covers a $(1-tc/n)$-fraction of all columns, the expectation of a single $X_i$ is $(tc/n)^{\tperm}$, therefore $\E[\sum_{i=1}^{n/c} X_i] = (n/c) \cdot (tc/n)^{\tperm}$ is much smaller than $t/c$. Therefore, we can use \cref{clm:NA_concentration} to bound the failure probability by
\[
\Pr\Bk*{\sum_{i=1}^{n/c} X_i > t/c} \leq \bk*{\frac{(n/c) \cdot (tc/n)^{\tperm}}{t/c}}^{\Omega(t/c)} = \bk*{c \cdot (tc/n)^{\Omega(\alpha \cdot c^2 \cdot \gamma)}}^{\Omega(t/c)},
\]
where we used the fact that $\tperm = \alpha \cdot c^2 \cdot \gamma$.
By a union bound over all subsets of augmented rows with size $\geq n/(2c)$, we show that the probability that any of them fails is
\begin{align*}
&\sum_{t=1}^{n/(2c)} \binom{n/c}{t} \bk*{c \cdot (tc/n)^{\Omega(\alpha \cdot c^2 \cdot \gamma)}}^{\Omega(t/c)}\\
\leq{} &\sum_{t=1}^{n/(2c)} (en/(tc))^t \bk*{c \cdot (tc/n)^{\Omega(\alpha \cdot c^2 \cdot \gamma)}}^{\Omega(t/c)} \tag{$\binom{n}{k} \leq \bk{\frac{en}{k}}^k$}\\
\leq{} &\sum_{t=1}^{n/(2c)} (tc/n)^{\Omega(\alpha \cdot c \cdot \gamma \cdot t)}\tag{$tc/n \leq 1/2$}\\
\leq{} &(n/(2c)) \max\BK*{(c/n)^{\Omega(\alpha \cdot c \cdot \gamma)}, (1/2)^{\Omega(\alpha \cdot n \cdot \gamma)}}\tag{$n/c \gg 1$}\\
\leq{} &1/n^{\Omega(\alpha \cdot \gamma)} = O(1/V),
\end{align*}
where the second-to-last line is derived using the convexity of $\log \bk{(tc/n)^{\alpha \cdot c \cdot \gamma \cdot t}}$ as a function of $t$, similarly to the small case;
the last line holds as we set $\alpha$ to a sufficiently large constant. Therefore, the probability that any subset of at least $n/(2c)$ augmented rows fails is $O(1/V)$, which concludes the proof of \cref{clm:perm_cover_large}.
\end{proof}

\subsection{Extensions of Theorem \ref{thm:augmented_retrieval_intro}}

\paragraph{Removing assumption of free randomness.}
\cref{thm:augmented_retrieval_intro} assumes free access to random hash functions and permutations. Now, we show how to remove this assumption at the cost of $O(n^{0.67})$ additional bits of space.

\begin{corollary}
    \label{col:augmented_retrieval_remove_random}
    Let $U=n^{1+O(1)}$, and let $V=n^{3+O(1)}$. For any constant $\eps> 0$, there is a static augmented retrieval data structure in the word RAM model with word size $\Theta(\log n)$ such that
    \begin{itemize}
    \item The data structure answers value queries for a set of $n - \eps n$ keys, where the keys are from a universe $[U]$ and values are from $[V]$.
    \item The data structure supports $\eps n$ augmented queries, which allow the user to access an array of $\eps n$ elements $a_1, \ldots, a_{\eps n} \in [V]$.
    \item The data structure answers any retrieval and augmented query in worst-case constant time, and uses $n\log V+O(n^{0.67})$ bits of space.
    \item The data structure does \emph{not} require access to external randomness.
    \end{itemize}
    
\end{corollary}

To remove the need for fully random hash functions and permutations, our algorithm finds a fixed set of hash functions and permutations that makes the construction succeed, and stores them explicitly as part of the data structure. In the original data structure, we need an $O(n)$-wise independent hash function and a constant number of random permutations over $[\eps n]$, which require $O(n)$ words to store. To reduce this space overhead, we apply the \emph{splitting trick} introduced in \cite{dietzfelbinger2009applicationsa}. The key idea is to hash the input retrieval queries into buckets of size $\tilde{O}(n^{2/3})$, allowing us to use $\tilde{O}(n^{2/3})$-wise independent hash functions for the construction within each bucket.

\begin{proof}
    Let $B=\tilde{O}(n^{2/3})$. The algorithm finds a set of hash functions and permutations that makes the following construction succeed, and stores them explicitly in the data structure.

    \begin{itemize}
        \item Use an $O(\poly \log n)$-wise independent hash function to distribute the retrieval queries into buckets such that each bucket receives $B$ retrieval queries in expectation. If any bucket contains more than $B+B/n^{1/3}$ retrieval queries, declare a construction failure. We also distribute the augmented queries evenly across the buckets.
        \item Sample an $O(B)$-wise independent hash function and a constant number of permutations of size $O(B)$, then share them across all buckets. Each bucket applies \cref{thm:augmented_retrieval_intro} to construct an augmented retrieval data structure with a fixed capacity $B+B/n^{1/3}$ of retrieval queries---regardless of the number of retrieval queries hashed into the bucket---and $\Theta(\eps B)$ augmented queries, using the shared hash functions and permutations. If any bucket's construction fails, declare a construction failure.
        \item Combine the results from each bucket and store a pointer to each bucket's data structure.
    \end{itemize}

    When distributing the keys, $O(\poly \log n)$-wise independence is sufficient to achieve Chernoff-type concentration bounds, as stated in \cite{schmidt1995chernoffhoeffding}. Therefore, with high probability, each bucket contains at most $B+B/n^{1/3}$ keys. With our hash functions and permutations, the overall construction succeeds with probability $1-O(n^2/V)$, since each individual bucket's construction succeeds with probability $1-O(n/V)$. This guarantees that the algorithm can find a set of hash functions and permutations for the construction to succeed.

    The redundancy of our data structure comes from four sources: the stored hash functions and permutations, the pointers to the buckets, the extra $B/n^{1/3}$ capacity added to each bucket to prevent overflow, and the redundancy incurred when applying \cref{thm:augmented_retrieval_intro}.

    \begin{itemize}
        \item The randomness we need consists of an $\tilde{O}(n^{2/3})$-wise independent hash function and a constant number of permutations of length $\tilde{O}(n^{2/3})$. These random permutations require $\tilde{O}(n^{2/3})$ bits to store. To simulate the $\tilde{O}(n^{2/3})$-wise independent hash function, we use Siegel's construction \cite{siegel1989universal}, which requires $O(n^{2/3 + \delta})$ bits to store for any constant $\delta > 0$. 
        \item The pointer to one bucket needs $O(\log n)$ bits, which amounts to $O(n^{1/3}\log n)$ bits in total.
        \item The total capacity we allocate to all buckets is $n + n^{2/3}$ values, incurring $\tilde{O}(n^{2/3})$ bits of redundancy compared to the $n$ values in the information-theoretic optimum.
        \item The construction in \cref{thm:augmented_retrieval_intro} incurs $O(\log^2 n)$ bits of redundancy for each bucket, which amounts to  $O(n^{1/3}\log^2 n)$ bits in total.
    \end{itemize}

    Therefore, the total redundancy of the data structure is $O(n^{0.67})$ bits.
\end{proof}

\paragraph{Handling small value universe.}

The augmented retrieval data structure in \cref{thm:augmented_retrieval_intro} only works for value universes larger than $n^3$. We now show that, by paying $o(n)$ bits of redundancy, we can allow the universe size to be as small as $\poly \log n$.

\begin{corollary}
    \label{col:augmented_retrieval_small_size}
    Let $U=n^{1+O(1)}$, and let $\log^{40} n\le V\leq \poly n$. For any constant $\eps> 0$, there is a static augmented retrieval data structure in the word RAM model with word size $\Theta(\log n)$ such that:
    \begin{itemize}
    \item The data structure answers value queries for a set of $n - \eps n$ keys, where the keys are from a universe $[U]$ and values are from $[V]$.
    \item The data structure supports $\eps n$ augmented queries, which allow the user to access an array of $\eps n$ elements $a_1, \ldots, a_{\eps n} \in [V]$.
    \item The data structure answers any retrieval and augmented query in worst-case constant time, and uses $n\log V+O(n^{0.67})+O(n/V^{1/40})$ bits of space.
    \item The data structure does \emph{not} require access to external randomness.
    \end{itemize}
\end{corollary}

When the value universe $V$ is small, our original data structure from \cref{thm:augmented_retrieval_intro} has a failure probability of $O(n^2/V)$, which becomes too high when $V$ is small. To address this, we again apply the splitting trick, hashing input retrieval queries into buckets of size smaller than $V^{1/3}$, which ensures that each bucket succeeds with constant probability. We further sample multiple sets of hash functions and permutations to boost the success probability in each bucket to $1 - 1/n^2$, then use a union bound over all buckets to bound the overall failure probability. Similar to \cref{col:augmented_retrieval_remove_random}, this splitting trick also removes the assumption of free randomness at a low cost.

\begin{proof}
    If $V \geq n^3$, the algorithm in \cref{thm:augmented_retrieval_intro} works directly. We therefore focus on the case where $V \le n^3$.

    Similar to \cref{col:augmented_retrieval_remove_random}, the algorithm finds a set of hash functions and permutations that make the following construction succeed, and stores them explicitly in the data structure:
    \begin{itemize}
        \item Use an $O(\poly \log n)$-wise independent hash function to distribute the retrieval queries into buckets such that each bucket receives $V^{1/4}$ retrieval queries in expectation. If any bucket receives more than $V^{1/4}+V^{1/5}$ queries, declare a construction failure. The algorithm also distributes the augmented queries evenly across these buckets.
        \item Sample $2 \log n$ groups of hash functions and permutations globally, each consisting of an $O(V^{1/4})$-wise independent hash function and a constant number of random permutations of appropriate length. For each bucket and each group, apply \cref{thm:augmented_retrieval_intro} to construct an augmented retrieval data structure with a fixed capacity of $V^{1/4}+V^{1/5}$ retrieval queries and $O(\eps V^{1/4})$ augmented queries.
        \item For each bucket, store the index of a group of hash functions and permutations (if any) that makes the construction succeed. If all groups make the construction fail for a certain bucket, declare a construction failure.
    \end{itemize}

    By concentration inequalities for limited independence \cite{schmidt1995chernoffhoeffding}, the probability that a bucket receives more than $V^{1/4}+V^{1/5}$ retrieval queries is at most $O(1/n^2)$ since $V \geq \log^{40} n$. With a specific group of hash functions and permutations, the construction in one bucket succeeds with probability $1 - O(V^{1/4}/V) \geq 1/2$. Therefore, for each bucket, with probability at least $1 - O(1/n^2)$, at least one of the $2 \log n$ groups will succeed. The overall success probability of the construction is thus $1 - O(1/n)$.

    The redundancy of our data structure comes from five sources:
    \begin{itemize}
        \item The stored hash functions and permutations: This consists of $O(\log n)$ instances of $O(V^{1/4})$-wise independent hash functions and permutations of length $O(V^{1/4})$. Using Siegel's construction \cite{siegel1989universal}, this requires $O(V^{1/4}n^{\delta})$ bits of space for any constant $\delta > 0$.
        \item The pointers to the buckets: Each pointer requires $O(\log n)$ bits, which amounts to $O(\log n \cdot n/V^{1/4})$ bits.
        \item The extra capacity in each bucket: The total capacity allocated to all buckets is $n + n/V^{1/20}$ values, incurring $O(\log V \cdot n/V^{1/20})$ bits of redundancy.
        \item The redundancy from \cref{thm:augmented_retrieval_intro}: Each bucket incurs $O(\log^2 n)$ bits of redundancy, which amounts to $O(\log^2 n \cdot n/V^{1/4})$ bits.
        \item The indices to the successful group: Each bucket needs an index in $[2\log n]$, which amounts to $O(\log \log n \cdot n/V^{1/4})$ bits.
    \end{itemize}
    Therefore, the overall redundancy is at most
    \[O(V^{1/4}n^{\delta})+O(\log^2 n \cdot n/V^{1/4})+O(\log V \cdot n/V^{1/20})\] bits. 
    Since $\log^{40} n\leq V\leq n^3$, the redundancy simplifies to $O(n/V^{1/40})$ bits.
\end{proof}

\section{Upper Bounds for Non-Augmented Retrieval}
\label{sec:non_augmented_ub}

In this section, we present a data structure for non-augmented retrieval that matches the lower bound \cref{thm:non_augmented_lb} when the value-universe size is a power of two.

\begin{restatable}{theorem}{ThmNonAugmentedRetrieval}
\label{thm:non_augmented_retrieval}
Let $v = O(\log n)$ be an integer and $U = \poly n$. There is a static retrieval data structure such that:
\begin{itemize}
    \item The data structure supports value queries for $n$ keys, where the keys are from a universe $[U]$ and values are $v$-bit integers. Each query takes $O(t)$ worst-case time for a parameter $t$, in the word RAM model with word size $w = \Theta(\log n)$.
    \item Letting $b = wt/v$, the data structure uses $nv + O(n/2^{\Omega(b)}+\log^2 n)$ bits of space.
    \item The data structure assumes free access to random hash functions. Over the randomness, the data structure succeeds with high probability in its construction process.
\end{itemize}
\end{restatable}

Our algorithm uses the construction in \cite{dietzfelbinger2019constanttime}, which can be viewed as a solution to the special case of \cref{thm:non_augmented_retrieval} where $v=1$.

\begin{theorem}[Proposition 3 of \cite{dietzfelbinger2019constanttime}]
\label{thm:non_augmented_matrix_construction}
    Let $U$ be the size of a key universe, and let $\ell$ be a parameter that is at least a large enough constant. In the word RAM model, assuming access to fully random hash functions, there is an algorithm that constructs a binary matrix $M$ of dimensions $U \times m$, where $m = n + n/2^{\Omega(\ell)}+O(\log n)$, such that
    \begin{itemize}
        \item The $m$ columns of $M$ are partitioned into blocks of size $\ell$. In each row of $M$, only two column blocks may contain non-zero entries. The positions of these two blocks and the value of the entries in them can be computed in time $O(\ell / w)$.
        \item For any given set of $n$ rows of $M$, the $n \times m$ submatrix formed by these rows has full row rank with constant probability.
    \end{itemize}
\end{theorem}

\newcommand{\Msub}{M_{\text{sub}}}

\begin{proof}[Proof of \cref{thm:non_augmented_retrieval}]
To develop intuition, we first introduce a warm-up data structure that uses $nv + nv/2^{\Omega(b)}+O(\log^2 n)$ bits of space, which will be improved to $nv + n/2^{\Omega(b)}+O(\log^2 n)$ bits later in the proof.
We follow the same matrix-solving framework as in \cref{sec:augmented_ub}.

\paragraph{A warm-up algorithm.} When $v = 1$, we view both the memory configuration of our data structure and the retrieval queries as vectors of length $n + n/2^{\Omega(b)}+O(\log^2 n)$, then answer each query by computing the inner product of the memory vector and the query vector. \cref{thm:non_augmented_matrix_construction} provides a way to construct a matrix $M$ of dimensions $U \times (n + n/2^{\Omega(b)}+O(\log n))$, where each row of $M$ corresponds to a query, such that any given set of $n$ rows of $M$ is linearly independent with constant probability. The construction algorithm then solves the linear equation to get the memory configuration that correctly answers all queries. We choose $\ell = b$ when applying \cref{thm:non_augmented_matrix_construction}, so the data structure uses $n + n/2^{\Omega(b)}+O(\log n)$ bits of space, and each query is answered in $O(b / w) = O(t)$ time. This construction matches the desired space and time bounds when $v = 1$.

When $v > 1$, we start with the same matrix $M$ of dimensions $U \times (n + n/2^{\Omega(b)}+O(\log n))$. To handle $v$-bit values, we solve $v$ separate linear systems, one for each bit position in the value. For the $i$-th bit position, we compute a memory vector $\mathbf{v}_i$ such that for each query $j$ with value $x_j$, we have that $\angbk{\mathbf{v}_i, M_j}$ equals the $i$-th bit of $x_j$.

To make queries efficient, we interleave the $v$ memory vectors $\mathbf{v}_1, \mathbf{v}_2, \ldots, \mathbf{v}_{v}$ into a single bit string $\mathbf{v}_{\text{all}}$ of length $(n + n/2^{\Omega(b)}+O(\log n)) \cdot v$ and store it in the memory. Specifically, for each position $i$ in the original vectors, we store the bits $\mathbf{v}_1(i), \mathbf{v}_2(i), \ldots, \mathbf{v}_{v}(i)$ consecutively in $\mathbf{v}_{\text{all}}$.

By \cref{thm:non_augmented_matrix_construction}, each row in matrix $M$ has non-zero entries only in two column blocks of length $b$. With our interleaved storage, when processing a query, we only need to access two segments of length $b \cdot v$ in the string $\mathbf{v}_{\text{all}}$. Since $b \cdot v = t \cdot w$, we can read these segments in $O(t)$ time, compute the inner products (with the help of a lookup table), and answer any query efficiently.

The space usage of this data structure is $nv + nv/2^{\Omega(b)}+O(v \log n)$ bits. When $b \ge \Omega(\log v)$ with a sufficiently large leading constant in $\Omega(\cdot)$, we have $v \le 2^{\Omega(b/2)}$, so the space usage is $nv + n/2^{\Omega(b/2)} + O(v \log n) = nv + n/2^{\Omega(b)}+O(\log^2 n)$ bits, as desired. In the following, we consider the regime where $b = O(\log v)$ and improve the space usage to $nv + n/2^{\Omega(b)}$ bits. In this regime, we have $n / 2^{\Omega(b)} \gg \log n$, so the matrix $M$ has $n + n / 2^{\Omega(b)}$ columns, where the $O(\log n)$ term is absorbed by $n / 2^{\Omega(b)}$.

\paragraph{Pivot and free columns.}

Let $\Msub$ be the submatrix of $M$ formed by the $n$ input retrieval queries. \cref{thm:non_augmented_matrix_construction} states that, with constant probability, $\Msub$ has full row rank. In this case, there exists a set $S$ of $n$ columns of $\Msub$ such that the $n \times n$ submatrix of $\Msub$ formed by the columns in $S$ is full rank. We fix a choice of $S$ and call the columns in $S$ the \defn{pivot} columns, and the columns in $\Msub \setminus S$ the \defn{free} columns. When we solve the linear systems, we additionally require that the entries in the free columns are zero in the memory vectors. With this requirement, there still exist feasible solutions to the linear systems, but we no longer need to store the bits of $\mathbf{v}_i$ in the free columns.

Formally, we define $\tilde{\mathbf{v}}_i$ to be the vector of $n$ bits obtained by removing the entries in the free columns from $\mathbf{v}_i$. We then interleave the vectors $\tilde{\mathbf{v}}_1, \tilde{\mathbf{v}}_2, \ldots, \tilde{\mathbf{v}}_v$ into a single bit string $\tilde{\mathbf{v}}_{\text{all}}$ of length $nv$ and store it in the memory. The only remaining task is to store the positions of the pivot/free columns to support fast queries.

\paragraph{Storing positions of the pivot columns.}
For each query, we need to know the positions of the pivot columns within the two column blocks where the query has non-zero entries. This requires both knowing the number of pivot columns to the left of a column block and knowing the specific positions of pivot columns within each block. To efficiently support these operations, we design the following data structures:

\begin{enumerate}
    \item We store a sparse bit vector of length $n + n/2^{\Omega(b)}$, with $n/2^{\Omega(b)}$ ones, supporting rank queries in constant time (e.g., using Theorem 2 of \cite{patrascu2008succincter}). Each bit indicates whether the corresponding column is a free column. The space used is \[\log\binom{n+n/2^{\Omega(b)}}{n/2^{\Omega(b)}}+\frac{n+n/2^{\Omega(b)}}{\poly\log n}+\tilde O\bk*{\bk*{n+n/2^{\Omega(b)}}^{3/4}}=O\bk*{n/2^{\Omega(b)}}\] bits. The latter terms are dominated by the first term since we assume that $b=O(\log\log n)$.
    \item We store a compact key-value dictionary where each key is the index of a column block containing at least one free column, and the associated value is a bitmask of length $b$ indicating which columns in this block are free. Since there are at most $n/2^{\Omega(b)}$ free columns, this dictionary contains at most $n/2^{\Omega(b)}$ entries. Therefore, it uses \[O\bk*{\log\binom{n+n/2^{\Omega(b)}}{n/2^{\Omega(b)}}+b\cdot n/2^{\Omega(b)}}=O\bk*{b\cdot n/2^{\Omega(b)}}=O\bk*{n/2^{\Omega(b)}}\] bits of space.
\end{enumerate}

Using these data structures, we can efficiently determine the positions of the pivot columns in any block. When processing a query, the algorithm first uses the rank data structure to locate the first pivot column in each relevant block, then queries the dictionary to get the exact positions of all pivot columns in these blocks. If a block is not in the dictionary, then all columns in that block are pivot columns.

Finally, after identifying the pivot columns, the algorithm reads the corresponding bits from $\tilde{\mathbf{v}}_{\text{all}}$, computes the inner products using a lookup table of size $O(n^\epsilon)$, and returns the answer. The total space usage is $nv + O(n/2^{\Omega(b)})$ bits, and the query time is $O(t)$.
\end{proof}

Similar to \cref{col:augmented_retrieval_remove_random}, we can remove the assumption of access to free randomness in this algorithm.
The proof of the following corollary is based on the splitting trick and is exactly the same as \cref{col:augmented_retrieval_remove_random}, so we omit it.

\begin{restatable}{corollary}{CorNonAugmentedRetrievalRemoveRandom}
    \label{col:non_augmented_retrieval_remove_random}
    Let $v = O(\log n)$ be an integer and $U = \poly n$. There is a static retrieval data structure such that:
    \begin{itemize}
    \item The data structure supports value queries for $n$ keys, where the keys are from a universe $[U]$ and values are $v$-bit integers.
    \item The data structure answers any retrieval query in $O(t)$ worst-case time for a parameter $t$, in the word RAM model with word size $w = \Theta(\log n)$.
    \item Letting $b = wt/v$, the data structure uses $nv + O(n^{0.67}) + O(n/2^{\Omega(b)})$ bits of space.
    \end{itemize}
\end{restatable}

Finally, we can construct filters from retrieval data structures using the well-known fingerprint method: Take a random hash function $h$; for a filter with $n$ keys $x_1,\cdots,x_n$ and false positive rate $\epsilon$, the algorithm constructs a retrieval data structure that maps key $x_i$ to $h(x_i)\in[\epsilon^{-1}]$. By using \cref{thm:non_augmented_retrieval} in this section, we get an algorithm for filters that uses $n\log\epsilon^{-1}+O(n/2^{\Omega(b)})$ bits assuming free randomness, which matches the lower bound in \cref{thm:non_augmented_lb_filter} when $\epsilon\geq n^2/U$. Therefore, for commonly used parameter regimes, the algorithm above is optimal. Furthermore, following the reduction in \cref{col:augmented_retrieval_remove_random}, we can also remove access to external randomness by paying $\tilde{O}(n^{2/3})$ bits of redundancy.

\begin{restatable}{corollary}{CorNonAugmentedFilterRemoveRandom}
    \label{col:non_augmented_filter_remove_random}
    Let $v = O(\log n)$ be an integer, $\epsilon = 2^{-v}$, and $U = \poly n$. There is a static filter such that:
    \begin{itemize}
    \item The filter can store $n$ keys from the universe $[U]$.
    \item The filter answers approximate membership queries with false positive rate $\epsilon$. Each query takes $O(t)$ worst-case time for a parameter $t$ in the word RAM model with word size $w = \Theta(\log n)$.
    \item Letting $b = wt / \log(1/\epsilon)$, the filter uses $n \log \epsilon^{-1} + O(n^{0.67}) + O(n/2^{\Omega(b)})$ bits of space.
    \end{itemize}
\end{restatable}

\section{Open Problems}
\label{sec:open}

We conclude the paper by listing several open questions that we believe are important to consider in future work:
\begin{enumerate}
    \item \textbf{\boldmath{}Static retrieval with values from a non-power-of-two domain $[V]$:} Theorem \ref{thm:retrieval_ub_lb} gives tight time/space bounds for static retrieval in the setting where values come from a power-of-two size domain $[V]$ (i.e., each value is $v = \log V$ bits). What is not clear is whether similar upper bounds can be achieved when $V$ is not a power of two. Is it possible to construct an $(n \log V + \lfloor n e^{-O(wt / \log V)}\rfloor)$-space solution with query time $t$, in the case where values are taken from an arbitrary domain $[V]$ satisfying $V \le \poly(n)$? The main obstacle here would be to provide a stronger version of Theorem \ref{thm:non_augmented_matrix_construction} (Proposition 3 in \cite{dietzfelbinger2019constanttime}) that would allow for values that are a non-integer number of bits. (We remark that, even in \cite{dietzfelbinger2019constanttime}, a variation of this question, in cases where $V$ is a valid field size, already appears as an open question.)
    \item \textbf{Small-value augmented retrieval: }Our upper bound for augmented retrieval (Corollaries \ref{col:augmented_retrieval_remove_random} and \ref{col:augmented_retrieval_small_size}) give redundancy $n^{1 - \Omega(1)}$ only if $V$ is relatively large---at least $n^{\Omega(1)}$. Whether or not one can achieve polynomially small redundancy $n^{1 - \Omega(1)}$ for $V = n^{o(1)}$ (even assuming, for example, that $V$ is a power of two) appears to remain as an open question, and appears to be a problem that will require significant additional algorithmic ideas to solve. 
    \item \textbf{\boldmath{}Augmented retrieval with $n^{\delta}$ redundancy: }Even allowing $V \ge n^3$, our upper bounds for augmented retrieval (namely, Corollary \ref{col:augmented_retrieval_remove_random}) still fall short of achieving a natural goal of $n^{\delta}$ redundancy, where $\delta \in (0, 1)$ is a small positive constant of our choice. The main obstacle here is the space needed to encode the hash functions used within the construction, which seem to require relatively high independence. Bringing this independence down to $n^\delta$, or finding some other avenue to reduce redundancy, remains an appealing direction for future work. 
    \item \textbf{Finding other problems that benefit from an augmented array: }The main takeaway in our paper is that, although static retrieval is not \emph{on its own} amenable to very time/space-efficient solutions, this barrier can be bypassed so long as we store the retrieval data structure \emph{together} with an augmented array. Are there other data-structural problems (besides retrieval data structures and filters) that can benefit from this high-level framework? 

    We remark that, in \cite{hu2025optimal}, augmented retrieval (albeit in a very restricted parameter regime) ends up being a critical \emph{algorithmic tool} for solving another data-structural problem (constructing a space-optimal static dictionary). If one can build a larger theory of ``augmented data structures'', it is likely that the results in this area would be quite useful as data-structural primitives that can be used within other (non-augmented) problems. 
\end{enumerate}

We remark that Questions 1 and 2 are especially important in the context of constructing \emph{static filters} (using retrieval data structures). In this setting, the false positive rate $\epsilon$ is typically quite large, meaning that the retrieval value-domain size $V = \epsilon^{-1}$ is typically quite small.

\bibliographystyle{alpha}
\bibliography{ref.bib}

\newcommand{\etalchar}[1]{$^{#1}$}
\begin{thebibliography}{LGM{\etalchar{+}}19}

\bibitem[AK21]{abdennebi2021bloom}
Anes Abdennebi and Kamer Kaya.
\newblock A {Bloom} filter survey: Variants for different domain applications.
\newblock Preprint arXiv:2106.12189, June 2021.

\bibitem[ANS10]{arbitman2010backyard}
Yuriy Arbitman, Moni Naor, and Gil Segev.
\newblock Backyard cuckoo hashing: Constant worst-case operations with a succinct representation.
\newblock In {\em Proc. 51st IEEE Symposium on Foundations of Computer Science (FOCS)}, pages 787--796, 2010.

\bibitem[BCF{\etalchar{+}}23]{bender2023tiny}
Michael~A. Bender, Alex Conway, Mart{\'i}n {Farach-Colton}, William Kuszmaul, and Guido Tagliavini.
\newblock Tiny pointers.
\newblock In {\em Proc. 34th ACM-SIAM Symposium on Discrete Algorithms (SODA)}, pages 477--508, 2023.

\bibitem[BCK{\etalchar{+}}14]{buhrman2014computing}
Harry Buhrman, Richard Cleve, Michal Kouck{\'y}, Bruno Loff, and Florian Speelman.
\newblock Computing with a full memory: Catalytic space.
\newblock In {\em Proc. 46th ACM Symposium on Theory of Computing (STOC)}, pages 857--866, New York, NY, USA, 2014.

\bibitem[BFJ{\etalchar{+}}12]{bender2012dont}
Michael~A. Bender, Mart{\'i}n {Farach-Colton}, Rob Johnson, Russell Kraner, Bradley~C. Kuszmaul, Dzejla Medjedovic, Pablo Montes, Pradeep Shetty, Richard~P. Spillane, and Erez Zadok.
\newblock Don't thrash: How to cache your hash on flash.
\newblock {\em Proceedings of the VLDB Endowment}, 5(11):1627--1637, July 2012.

\bibitem[BFK{\etalchar{+}}22]{bender2022optimal}
Michael~A. Bender, Mart{\'i}n {Farach-Colton}, John Kuszmaul, William Kuszmaul, and Mingmou Liu.
\newblock On the optimal time/space tradeoff for hash tables.
\newblock In {\em Proc. 54th ACM SIGACT Symposium on Theory of Computing (STOC)}, pages 1284--1297, 2022.

\bibitem[BKZ25]{bender2025optimal}
Michael~A. Bender, William Kuszmaul, and Renfei Zhou.
\newblock Optimal non-oblivious open addressing.
\newblock In {\em Proc. 57th ACM Symposium on Theory of Computing (STOC)}, pages 268--277, 2025.

\bibitem[Blo70]{bloom1970space}
Burton~H. Bloom.
\newblock Space/time trade-offs in hash coding with allowable errors.
\newblock {\em Communications of the ACM}, 13(7):422--426, July 1970.

\bibitem[CFG{\etalchar{+}}78]{carter1978exact}
Larry Carter, Robert Floyd, John Gill, George Markowsky, and Mark Wegman.
\newblock Exact and approximate membership testers.
\newblock In {\em Proc. 10th ACM Symposium on Theory of Computing (STOC)}, pages 59--65, 1978.

\bibitem[Cle84]{clerry1984compact}
John~G. Clerry.
\newblock Compact hash tables using bidirectional linear probing.
\newblock {\em IEEE Transactions on Computers}, 33(9):828--834, September 1984.

\bibitem[CLJW17]{chen2017dynamic}
Hanhua Chen, Liangyi Liao, Hai Jin, and Jie Wu.
\newblock The dynamic cuckoo filter.
\newblock In {\em Proc. 25th IEEE International Conference on Network Protocols (ICNP)}, pages 1--10, 2017.

\bibitem[DHSW22]{dillinger2022fast}
Peter~C. Dillinger, Lorenz {H{\"u}bschle-Schneider}, Peter Sanders, and Stefan Walzer.
\newblock Fast succinct retrieval and approximate membership using {Ribbon}.
\newblock In {\em Proc. 20th International Symposium on Experimental Algorithms (SEA)}, pages 4:1--4:20, 2022.

\bibitem[DMPP06]{demaine2006dictionariis}
Erik~D. Demaine, Friedhelm {Meyer auf der Heide}, Rasmus Pagh, and Mihai P{\v a}tra{\c s}cu.
\newblock De dictionariis dynamicis pauco spatio utentibus.
\newblock In {\em Proc. 7th Latin American Conference on Theoretical Informatics (LATIN)}, pages 349--361, 2006.

\bibitem[DP08]{dietzfelbinger2008succinct}
Martin Dietzfelbinger and Rasmus Pagh.
\newblock Succinct data structures for retrieval and approximate membership.
\newblock In {\em Proc. 35th International Colloquium on Automata, Languages and Programming (ICALP)}, pages 385--396, 2008.

\bibitem[DPT10]{dodis2010changing}
Yevgeniy Dodis, Mihai P{\v a}tra{\c s}cu, and Mikkel Thorup.
\newblock Changing base without losing space.
\newblock In {\em Proc. 42nd ACM Symposium on Theory of Computing (STOC)}, pages 593--602, 2010.

\bibitem[DR09]{dietzfelbinger2009applicationsa}
Martin Dietzfelbinger and Michael Rink.
\newblock Applications of a splitting trick.
\newblock In {\em Proc. 36th International Colloquium on Automata, Languages and Programming (ICALP): Part I}, pages 354--365, 2009.

\bibitem[DW19]{dietzfelbinger2019constanttime}
Martin Dietzfelbinger and Stefan Walzer.
\newblock Constant-time retrieval with {$O(\log m)$} extra bits.
\newblock In {\em Proc. 36th International Symposium on Theoretical Aspects of Computer Science (STACS)}, pages 24:1--24:16, 2019.

\bibitem[FAKM14]{fan2014cuckoo}
Bin Fan, David~G. Andersen, Michael Kaminsky, and Michael Mitzenmacher.
\newblock Cuckoo filter: Practically better than {Bloom}.
\newblock In {\em Proc. 10th ACM International on Conference on emerging Networking EXperiments and Technologies (CoNEXT)}, pages 75--88, 2014.

\bibitem[FK84]{fredman1984size}
Michael~L. Fredman and J{\'a}nos Koml{\'o}s.
\newblock On the size of separating systems and families of perfect hash functions.
\newblock {\em SIAM Journal on Algebraic Discrete Methods}, 5(1):61--68, March 1984.

\bibitem[GFO18]{geil2018quotient}
Afton Geil, Mart{\'i}n {Farach-Colton}, and John~D. Owens.
\newblock Quotient filters: Approximate membership queries on the {GPU}.
\newblock In {\em Proc. IEEE International Parallel and Distributed Processing Symposium (IPDPS)}, pages 451--462, 2018.

\bibitem[GL20]{graf2020xor}
Thomas~Mueller Graf and Daniel Lemire.
\newblock Xor filters: Faster and smaller than {Bloom} and cuckoo filters.
\newblock {\em ACM Journal of Experimental Algorithmics}, 25, March 2020.

\bibitem[HLY{\etalchar{+}}25]{hu2025optimal}
Yang Hu, Jingxun Liang, Huacheng Yu, Junkai Zhang, and Renfei Zhou.
\newblock Optimal static dictionary with worst-case constant query time.
\newblock In {\em Proc. 57th ACM Symposium on Theory of Computing (STOC)}, pages 278--289, 2025.

\bibitem[HT01]{hagerup2001efficient}
Torben Hagerup and Torsten Tholey.
\newblock Efficient minimal perfect hashing in nearly minimal space.
\newblock In {\em Proc. 18th Symposium on Theoretical Aspects of Computer Science (STACS)}, pages 317--326, 2001.

\bibitem[KG14]{krause2014submodular}
Andreas Krause and Daniel Golovin.
\newblock Submodular function maximization.
\newblock In {\em Tractability: Practical Approaches to Hard Problems}, pages 71--104. Cambridge University Press, Cambridge, 2014.

\bibitem[KPX{\etalchar{+}}25]{kuszmaul2025tight}
William Kuszmaul, Aaron Putterman, Tingqiang Xu, Hangrui Zhou, and Renfei Zhou.
\newblock Tight bounds and phase transitions for incremental and dynamic retrieval.
\newblock In {\em Proc. 36th ACM-SIAM Symposium on Discrete Algorithms (SODA)}, pages 3974--3997, 2025.

\bibitem[KW24]{kuszmaul2024space}
William Kuszmaul and Stefan Walzer.
\newblock Space lower bounds for dynamic filters and value-dynamic retrieval.
\newblock In {\em Proc. 56th ACM Symposium on Theory of Computing (STOC)}, pages 1153--1164, 2024.

\bibitem[LGM{\etalchar{+}}19]{luo2019optimizing}
Lailong Luo, Deke Guo, Richard T.~B. Ma, Ori Rottenstreich, and Xueshan Luo.
\newblock Optimizing {Bloom} filter: Challenges, solutions, and comparisons.
\newblock {\em IEEE Communications Surveys \& Tutorials}, 21(2):1912--1949, 2019.

\bibitem[LGR{\etalchar{+}}19]{luo2019consistent}
Lailong Luo, Deke Guo, Ori Rottenstreich, Richard T.~B. Ma, Xueshan Luo, and Bangbang Ren.
\newblock The consistent cuckoo filter.
\newblock In {\em Proc. IEEE Conference on Computer Communications (INFOCOM)}, pages 712--720, 2019.

\bibitem[LP13]{lovett2013space}
Shachar Lovett and Ely Porat.
\newblock A space lower bound for dynamic approximate membership data structures.
\newblock {\em SIAM Journal on Computing}, 42(6):2182--2196, 2013.

\bibitem[Meh84]{mehlhorn1984data}
Kurt Mehlhorn.
\newblock {\em Data Structures and Algorithms 3: Multi-Dimensional Searching and Computational Geometry}.
\newblock Springer-Verlag, Berlin, Heidelberg, 1984.

\bibitem[MPP05]{mortensen2005dynamic}
Christian~Worm Mortensen, Rasmus Pagh, and Mihai P{\v a}tra{\c s}cu.
\newblock On dynamic range reporting in one dimension.
\newblock In {\em Proc. 37th ACM Symposium on Theory of Computing (STOC)}, pages 104--111, 2005.

\bibitem[MR95]{motwani1995randomized}
Rajeev Motwani and Prabhakar Raghavan.
\newblock {\em Randomized Algorithms}.
\newblock Cambridge University Press, Cambridge, 1995.

\bibitem[P{\v a}t08]{patrascu2008succincter}
Mihai P{\v a}tra{\c s}cu.
\newblock Succincter.
\newblock In {\em Proc. 49th IEEE Symposium on Foundations of Computer Science (FOCS)}, pages 305--313, 2008.

\bibitem[PBJP17]{pandey2017generalpurpose}
Prashant Pandey, Michael~A. Bender, Rob Johnson, and Rob Patro.
\newblock A general-purpose counting filter: Making every bit count.
\newblock In {\em Proc. ACM International Conference on Management of Data (SIGMOD)}, pages 775--787, 2017.

\bibitem[PCD{\etalchar{+}}21]{pandey2021vector}
Prashant Pandey, Alex Conway, Joe Durie, Michael~A. Bender, Mart{\'i}n {Farach-Colton}, and Rob Johnson.
\newblock Vector quotient filters: Overcoming the time/space trade-off in filter design.
\newblock In {\em Proc. International Conference on Management of Data (SIGMOD)}, pages 1386--1399, New York, NY, USA, 2021.

\bibitem[Por09]{porat2009optimal}
Ely Porat.
\newblock An optimal {Bloom} filter replacement based on matrix solving.
\newblock In {\em Proc. 4th International Computer Science Symposium in Russia (CSR)}, pages 263--273, 2009.

\bibitem[PPS05]{pagh2005optimal}
Anna Pagh, Rasmus Pagh, and Srinivasa~Rao Satti.
\newblock An optimal {Bloom} filter replacement.
\newblock In {\em Proc. 16th ACM-SIAM Symposium on Discrete Algorithms (SODA)}, pages 823--829, 2005.

\bibitem[PR04]{pagh2004cuckoo}
Rasmus Pagh and Flemming~Friche Rodler.
\newblock Cuckoo hashing.
\newblock {\em Journal of Algorithms}, 51(2):122--144, May 2004.

\bibitem[QL92]{qi1992chebychevs}
Yao Qi and Shituo Lou.
\newblock A {Chebychev's} type of prime number theorem in a short interval {II}.
\newblock {\em Hardy--Ramanujan Journal}, 15:1--33, 1992.

\bibitem[Sie89]{siegel1989universal}
Alan Siegel.
\newblock On universal classes of fast high performance hash functions, their time-space tradeoff, and their applications.
\newblock In {\em Proc. 30th IEEE Symposium on Foundations of Computer Science (FOCS)}, pages 20--25, 1989.

\bibitem[SSS95]{schmidt1995chernoffhoeffding}
Jeanette~P. Schmidt, Alan Siegel, and Aravind Srinivasan.
\newblock {Chernoff-Hoeffding} bounds for applications with limited independence.
\newblock {\em SIAM Journal on Discrete Mathematics}, 8(2):223--250, May 1995.

\bibitem[Von10]{vondrak2010note}
Jan Vondrak.
\newblock A note on concentration of submodular functions.
\newblock Preprint arXiv:1005.2791, May 2010.

\bibitem[Waj17]{wajc2017negative}
David Wajc.
\newblock Negative association -- definition, properties, and applications.
\newblock Available online at \url{https://www.cs.cmu.edu/~dwajc/notes/Negative%20Association.pdf}, 2017.

\bibitem[Yu20]{yu2020nearly}
Huacheng Yu.
\newblock Nearly optimal static {Las} {Vegas} succinct dictionary.
\newblock In {\em Proc. 52nd ACM SIGACT Symposium on Theory of Computing (STOC)}, pages 1389--1401, 2020.

\end{thebibliography}

\appendix

\section{Proofs of Lower Bound Extensions}
\label{app:lb}
In this appendix, we provide the proofs for \cref{thm:augmented_lb} and \cref{thm:non_augmented_lb_filter}.

\subsection{Proof of Theorem \ref{thm:augmented_lb}}

\AugmentedLB*

\begin{proof}
    This proof analyzes a slightly modified version of the communication game in the proof of \cref{thm:non_augmented_lb}, where instead of $X$ and $A$, Alice also needs to send the array of augmented elements $B=(b_1,\dots,b_m)\in [V]$.
    
    As in the proof of \cref{thm:non_augmented_lb}, we let $\beta$ be a sufficiently large constant, and aim to prove a space lower bound of $(n+m)v + \lfloor n \cdot e^{-\beta wt/v} \rfloor$ bits, when $m\le n\cdot e^{-\beta wt/v}$. We only consider the case where $wt/v \in [1, (1/\beta) \ln n]$. Assume for the sake of contradiction that the data structure uses $(n+m)v+R$ bits of space, where $R\le n\cdot e^{-\beta wt/v}$.

    Define $\tslow, \ft,\pb, f,T,M$ exactly as in the proof of \cref{thm:non_augmented_lb}. The only difference from the previous proof is that, the number of cells $M$ is now $((n+m)v+R)/w$.

    \paragraph{The protocol.} Alice sends one bit, indicating whether \cref{prop:avg_time_bound} holds. If not, Alice directly sends $X$, $A$, and $B$, and terminates.
    
    If \cref{prop:avg_time_bound} holds, then Alice sends another bit, indicating whether $f(X)\ge T$.

    \paragraph{Case 1: $f(X)<T$.} In this case, Alice first sends the exact same message as in case 1 of \cref{thm:non_augmented_lb}, then explicitly sends the augmented elements $B$ using $\ceil{mv}$ bits. The overall message size is less than $\log\binom{U}{n}+(n+m)v-10$.

    \paragraph{Case 2: $f(X)\ge T$.} In this case, Alice first sends the exact same message as in case 2 of \cref{thm:non_augmented_lb}, i.e., Alice first sends the data structure $D$, and then sends $X$ using an efficient encoding. Note that after receiving $D$, Bob can simulate all the queries and recover the augmented elements $B$, so we do not need to explicitly send those elements.

    To show that this protocol is efficient, we prove that the number of sets $X$ such that $f(X)\ge T$ is small. Let $S\subset [U]$ be a random set, defined as in \cref{thm:non_augmented_lb}. We show that when $m$ is small, we can still bound the probability that $f(S)\ge T$.

    On the one hand, $\E[f(S)]$ does not change, and is at most $M(1-e^{-20wt/v})$. On the other hand,
    \begin{align*}
        T&=(nv-(10t/\tslow)nv-1)/w \\
        &=M-(mv+R+(10t/\tslow)nv+1)/w \\
        &\ge M-mv/w-(nv/w)\cdot e^{-(\beta/3)wt/v}\tag{from the proof of Theorem \ref{thm:non_augmented_lb}}\\
        &\ge M-(nv/w)\cdot e^{-\beta wt/v}-(nv/w)\cdot e^{-(\beta/3)wt/v}\\
        &\ge M(1-e^{-(\beta/4)wt/v}).
    \end{align*}
    Therefore the concentration inequality still applies, and the rest of the proof is identical to \cref{thm:non_augmented_lb}.
\end{proof}

\subsection{Proof of Theorem \ref{thm:non_augmented_lb_filter}}

\FilterLB*

\begin{proof}
    We begin with the classical information-theoretic lower bound for filters (see \cite{dietzfelbinger2008succinct}), which implies that any filter must use at least
    \begin{align*}
        H\defeq \log\binom{U}{n}-\log\binom{n+(U-n)\epsilon}{n}
    \end{align*}bits even without time constraints (note that when $\epsilon\ge n^2/U$, we have $H=n\log \eps^{-1}-O(1)$). We then strengthen this bound to $H + \lfloor n \cdot e^{-O(wt/\log \eps^{-1})} \rfloor$ bits using techniques of \cref{thm:non_augmented_lb}.

    \paragraph{The information-theoretic lower bound.} The argument of \cite{dietzfelbinger2008succinct} uses a one-way communication game in which Alice sends a set $X\subset[U]$ of size $n$ to Bob. Assume, for contradiction, that a filter uses fewer than $H$ bits; then there exists a protocol that communicates fewer than $\log\binom{U}{n}$ bits, a contradiction.
    
    The protocol is as follows: Instead of sending $X$ directly, Alice first builds a filter $D$ for $X$ and sends $D$ to Bob (which uses fewer than $H$ bits by assumption). This filter $D$ shrinks the potential elements of $X$ to a set $Y$ consisting of keys for which the filter answers true. Bob can simulate all queries on $D$ to obtain $Y$, and then Alice needs only $\log\binom{|Y|}{n}$ additional bits to specify $X$ as a subset of $Y$. By definition of filters, the expected size of $Y$ is at most $n+(U-n)\epsilon$. Hence, the expected communication is at most
    \begin{align*}
        H + \E\Bk*{\log\binom{|Y|}{n}}\le H +\log\binom{n+(U-n)\epsilon}{n}
    \end{align*}
    bits. Here we use the fact that $\log\binom{x}{n}$ is a concave function of $x$.%
    \footnote{To see this, note that $\log\binom{x}{n}$ can be written as a summation $\bk{\sum_{i=0}^{n-1}\log(x-i)}-\log n!$, where each $\log(x-i)$ is concave.}
    By definition of $H$, the expected communication cost is less than $\log\binom{U}{n}$ bits, leading to a contradiction. Therefore, the space usage of a filter must be at least $H$ bits.
    
    \paragraph{Overview of the improved lower bound.}
    To obtain the lower bound stated in \cref{thm:non_augmented_lb_filter}, we further improve the previous communication protocol by following the techniques of \cref{thm:non_augmented_lb}: 
    If the queries on the keys in $X$ probe few cells (i.e., the case $f(X)<T$ in \cref{thm:non_augmented_lb}), then instead of sending the entire filter $D$, Alice only needs to send the addresses and contents of those probed cells, which saves the communication. If the queries on the keys in $X$ probe many cells (i.e., the case $f(X)\ge T$ in \cref{thm:non_augmented_lb}), then, after Alice sends the filter $D$ and Bob recovers the set $Y$ as the previous protocol, because a random subset $S\subseteq Y$ of size $n$ is unlikely to have this property, Alice can send $X$ more efficiently conditioned on the event $f(X) \ge T$. Combining the two cases yields a smaller communication cost and thus a stronger lower bound.

    However, a new technical challenge arises: the ``universe'' of $X$ is now the random set $Y$ rather than the fixed set $[U]$. In some rare realizations of $Y$, the above approach yields no savings. For example, if $|Y|<2n$, then $\log \binom{|Y|}{n}=O(n)$ is already very small, leaving little room for savings (and the argument in \cref{thm:non_augmented_lb} assumes $U\ge 2n$).
    Moreover, if the keys in $Y$ have large average expected query time, a random subset $S\subseteq Y$ of size $n$ may happen to satisfy $f(S)\ge T$ with non-negligible probability, making the $f(X)\ge T$ branch ineffective.
    To handle this, we use two tailored strategies. When $|Y|<2n$, we simply send $X$ using $\log\binom{|Y|}{n}$ bits. When $Y$ has large average query time, we apply a different concentration bound to show that a random subset $S\subseteq Y$ is still unlikely to equal $X$.

    \paragraph{Notations.} Let $\beta$ be a sufficiently large constant multiple of $\log_n U$. Our goal is to prove a space lower bound of $H + \lfloor n \cdot e^{-\beta wt/\log\epsilon^{-1}} \rfloor$ bits.
    We restrict attention to $wt/\log\epsilon^{-1} \in [1, (1/\beta) \ln n]$. For contradiction, assume the filter uses $H+R$ bits, where $R< n\cdot e^{-\beta wt/\log\epsilon^{-1}}$.

    Define $\tslow,\ft,\pb, f$ exactly as in the proof of \cref{thm:non_augmented_lb}, but for the truncated filter queries. Let $g(S)\defeq\sum_{x\in S}f(x)$ denote the sum of truncated query times.
    
    Let $M\defeq(H+R)/w$ denote the number of cells in $D$. Note that $\log M\le w$, since $M\le 2H/w\le 2n\le U\le 2^w$. We will also define a threshold $T$ with its value differing from \cref{thm:non_augmented_lb} and will be chosen later.

    Let $P$ be a random variable, denoting the number of keys in the universe for which the filter returns a false positive when queried. We have that $\E[P]=(U-n)\epsilon$.

    \paragraph{The protocol.} As in \cref{thm:non_augmented_lb}, we first rule out bad cases: either $P<n$ (i.e., $|Y|<2n$), or the average truncated query time of $X$ exceeds $10t$ (formally, $g(X)>10nt$).

    Alice sends one bit indicating whether we are in a bad case. If so, she sends the filter using $H+R$ bits, then sends $X$ using $\log\binom{n+P}{n}$ bits. Otherwise, we proceed as in \cref{thm:non_augmented_lb}, as follows.
    
    In the good case, Alice first sends one bit, indicating whether $f(X)\ge T$.

    \paragraph{Case 1: $f(X)< T$.} We show that it is possible to ``send $D$'' using fewer than $H+R$ bits (although we are not really sending $D$ in this case, but some messages that play the role of $D$). Alice sends the addresses and contents of the cells in $\pb(X)$ to Bob, using $\lceil\log T+\log \binom{M}{f(X)}+f(X)\cdot w\rceil$ bits (first the size of $\pb(X)$ in $\log T$ bits, then the addresses, then the contents). The number of bits sent is at most
    \begin{align*}
        &\ceil*{\log T+\log \binom{M}{f(X)}+f(X)\cdot w} \\
        \le{}&\log \bk*{\bk*{\frac{eM}{M-f(X)}}^{M-f(X)}}+f(X)\cdot w+1+\log T\tag{$\binom{n}{k}\le \bk{\frac{en}{k}}^k$} \\
        \le{}&(M-f(X))\cdot \bk{w+\log e - \log \bk{M-f(X)}}+f(X)\cdot w+1+\log T\tag{$\log M\le w$} \\
        \le{}&M\cdot w-(M-T)\cdot(\log (M-T)-\log e)+1+\log T\tag{$f(X)<T$}\\
        ={} &H+R-(M-T)\cdot(\log (M-T)-\log e)+1+\log T.
    \end{align*}
    
    Upon receiving these cells, Bob simulates all queries: a query returns false if it probes any cell outside $\pb(X)$ or if it is slow. He obtains a candidate set $\tilde{Y}$ of keys $x$ for which $\textsf{query}(x)$ is fast and returns true. Note that $\tilde{Y}$ contains the fast queries in $X$ and has size $\le n+P$. 
    
    Given $\tilde{Y}$, Alice sends $\ceil{\log \binom{n+P}{n}}$ additional bits to specify the fast queries in $X$. She then explicitly sends the slow queries in $X$, using at most $\ceil{(10t/\tslow)n\cdot \log U}$ bits, since there are at most $(10t/\tslow)n$ of them.

    To sum up, Alice's message consists of the following:

    \begin{enumerate}
        \item One bit indicating that we are in a good case;
        \item One bit indicating that $f(X)<T$;
        \item The address and content of $\pb(X)$, using $H+R-(M-T)\cdot (\log (M-T)-\log e)+1+\log T$ bits;
        \item The identity of the fast queries in $X$, using $\ceil{\log\binom{n+P}{n}}$ bits.
        \item The identity of the slow queries in $X$, using $\ceil{(10t/\tslow)n\cdot \log U}$ bits.
    \end{enumerate}

    The total number of bits is at most
    \begin{align*}
        &H+\log\binom{n+P}{n}-(M-T)\cdot (\log (M-T)-\log e)+R+(10t/\tslow)n\cdot \log U+\log T+10. \numberthis \label{ineq:cost_case1}
    \end{align*}

    We choose $T$ so that \eqref{ineq:cost_case1} is at most $H+\log\binom{n+P}{n}-2R-10$. Equivalently,
    \begin{align*}
        (M-T)\cdot(\log (M-T)-\log e)>3R+(10t/\tslow)n\cdot \log U+\log T+20. \numberthis \label{ineq:cost_case1_simplified}
    \end{align*}
    Since
    \begin{align*}
        &3R+(10t/\tslow)n\cdot \log U+\log T+20 \\
        \le {}&n\cdot e^{-(\beta/2)wt/\log \epsilon^{-1}}+10(t/\tslow)n\cdot \log U\tag{$\log T\le \log n$} \\
        = {}&n\cdot e^{-(\beta/2)wt/\log \epsilon^{-1}}+10 n\cdot \log U\cdot e^{-(\beta/2)wt/\log \epsilon^{-1}} \tag{$\tslow\defeq t\cdot e^{(\beta/2) wt/\log \epsilon^{-1}}$}\\
        \le {}&n\cdot \log n\cdot e^{-(\beta/3)wt/\log\epsilon^{-1}}\tag{$\beta\gg \log U/\log n$}
    \end{align*}
    Thus, solving \eqref{ineq:cost_case1_simplified}, we may set $T:=M-n\cdot e^{-(\beta/4)wt/\log \epsilon^{-1}}$, which makes the total communication cost in this case at most $H+\log\binom{n+P}{n}-2R-10$.

    \paragraph{Case 2: $f(X)\ge T$.} In this case, we show that it is possible to send $X$ using fewer than $\log\binom{n+P}{n}$ bits. Alice's message consists of the following:

    \begin{enumerate}
        \item One bit indicating that we are in a good case;
        \item One bit indicating that $f(X)\ge T$;
        \item The memory configuration $D$, using $H+R$ bits;
        \item The set $X$ conditioned on $D$, the event $f(X)\ge T$, and the fact that we are in a good case.
    \end{enumerate}

    We now bound the length of the last message. We will show that, a random subset $S$ of $Y$ is unlikely to satisfy both $f(S)\ge T$ and $g(S)\le 10nt$, which means that Alice can send $X$ more efficiently conditioned on these events.

    Specifically, define $S$ to be a random subset of $Y$, where each element is included independently with probability $n/|Y|$. Then the cost of this step is at most 
    \begin{align*}
        &-\log\Pr[S=X|f(S)\ge T,g(S)\le 10nt] \\
        \le& -\log\Pr[S=X]+\min\{\log\Pr[f(S)\ge T],\log\Pr[g(S)\le 10nt]\}.
    \end{align*}
    In the proof of \cref{thm:non_augmented_lb}, we showed that $-\log\Pr[S=X]=\log\binom{|Y|}{n}+O(\log n)\le \log\binom{n+P}{n}+O(\log n)$. It remains to show that the second term is small. We consider two cases based on whether the average truncated query time over $Y$ is large. The protocol is the same in both cases; only the analyses differ.

    \paragraph{Case 2.1: $g(Y)\le 100t\cdot |Y|$.}  Similarly to \cref{clm:efs_small}, we can show that $\E[f(S)] \le M (1 - e^{-200wt/\log\epsilon^{-1}})$. Note that \cref{clm:efs_small} implicitly requires that $U\ge 2n$, which in our setting translates to $|Y|\ge 2n$ (since $S$ is a random subset of $Y$), and this is satisfied because $|Y|=n+P$ and $P\ge n$ (we are in the good case). 
    
    In comparison, we have that
    \begin{align*}
        T&=M-n\cdot e^{-(\beta/4)wt/\log\epsilon^{-1}} \\
        &\ge M-M\cdot e^{-(\beta/5)wt/\log\epsilon^{-1}} \tag{$M\ge \bk{n\log\epsilon^{-1}-O(1)}/w$} \\
        &=M(1-e^{-(\beta/5)wt/\log\epsilon^{-1}}),
    \end{align*}
    thus there is a gap between $\E[f(S)]$ and $T$. We can therefore apply a concentration bound to show that $\Pr[f(S)\ge T] \le \exp\bk*{-n\cdot e^{-(3\beta/4)wt/\log\epsilon^{-1}}}$, as in \cref{thm:non_augmented_lb}.
    
    \paragraph{Case 2.2: $g(Y)> 100t\cdot |Y|$.} In this case, we bound $\Pr[g(S)\le 10nt]$ via another concentration bound. Since $g(Y)>100t\cdot |Y|$, we have $\E[g(S)]>100nt$. As $g(S)$ is the sum of independent random variables in $[0,\tslow]$, a standard Chernoff bound yields
    \begin{align*}
        \Pr[g(S)\le 10nt]&\le\Pr[g(S)\le (1-0.9)\E[g(S)]] \\
        &\le \bk*{\frac{e^{-0.9}}{0.1^{0.1}}}^{\E[g(S)]/\tslow} \\
        &\le \exp\bk{-0.1\cdot 100nt/(t\cdot e^{(\beta/2)wt/\log\epsilon^{-1}})} \\
        &\le \exp\bk{-n\cdot e^{-(2\beta/3)wt/\log\epsilon^{-1}}}.
    \end{align*}

    In both subcases 2.1 and 2.2, the cost of the last message is at most $\log\binom{n+P}{n}-n\cdot e^{-(3\beta/4)wt/\log\epsilon^{-1}}$ (deterministically).

    \paragraph{Putting the pieces together.} To summarize what we have so far:
    \begin{itemize}
        \item In a bad case, we use $H+\log\binom{n+P}{n}+R+1$ bits.
        \item In a good case, we use at most $H+\log\binom{n+P}{n}-2R-10$ bits.
    \end{itemize}

    It remains to show that, for any distribution of $P$ with $\E[P]= (U-n)\epsilon$, the expected total cost of the protocol is less than $\log\binom{U}{n}$ bits.

    Intuitively, if the bad case is rare, the loss of $R+1$ bits there is offset by the savings of $2R+10$ bits in the good case, and we can apply the same Jensen's inequality as before to show the expected communication cost is less than $H+\log \binom{n+\E[P]}{n} \le \log \binom{U}{n}$. If the bad case is frequent, then $P$ has substantial mass below $n$ (i.e., it is not concentrated around its mean), which implies that $\E\Bk*{\log\binom{n+P}{n}}$ is much smaller than $\log\binom{n+\E[P]}{n}$. Formally, we consider the following two cases:

    \begin{itemize}
        \item $\Pr[P<n]<1/10$: The bad-case probability is at most $\Pr[P<n]+\Pr[g(X)>10nt]\le 1/5$ (since $\E[g(X)]\le nt$ implies $\Pr[g(X)>10nt]\le 1/10$ by Markov). Therefore,
        \begin{align*}
            \E\Bk*{\text{cost}}&{}\le n\log\epsilon^{-1}+\E\Bk*{\log\binom{n+P}{n}}+\Pr\Bk*{\text{bad case}}\cdot (R+1)-\Pr\Bk*{\text{good case}}\cdot (2R+10) \\
            &{}<n\log\epsilon^{-1}+\E\Bk*{\log\binom{n+P}{n}}\le n\log\epsilon^{-1}+\log\binom{n+\E[P]}{n}\le \log\binom{U}{n}.
        \end{align*}
        \item $\Pr[P<n]\ge 1/10$: Here $\E\Bk*{\log\binom{n+P}{n}}$ is much smaller than $\log\binom{n+\E[P]}{n}$. Let $p:=\Pr[P<n]$. Then
        \begin{align*}
            \E\Bk*{\log\binom{n+P}{n}}&\le \Pr[P<n]\cdot \log\binom{2n}{n}+\Pr[P\ge n]\cdot \E\Bk*{\log\binom{n+P}{n}\biggr|P\ge n} \\
            &\le p\cdot n+(1-p)\cdot \log\binom{n+\E[P|P\ge n]}{n} \\
            &\le p\cdot n+(1-p)\cdot \log\binom{n+\E[P]/(1-p)}{n} \\
            &\le p\cdot n+(1-p)\cdot \bk*{\log\binom{n+\E[P]}{n}-n\log(1-p)},
        \end{align*}
        where the last line is because
        \begin{align*}
            \binom{n+\E[P]/(1-p)}{n}&=\frac 1{n!}\prod_{i=1}^{n}\bk*{\frac {\E[P]}{1-p}+i} 
            \le \frac 1{n!}\prod_{i=1}^{n}\frac{\E[P]+i}{1-p} 
            =\binom{n+\E[P]}{n}\cdot \bk*{\frac 1{1-p}}^n.
        \end{align*}
        Since $-(1-p)\log(1-p)<1$ for $p\in [0,1]$, we obtain
        \begin{align*}
            \E\Bk*{\log\binom{n+P}{n}}&\le p\cdot n+n+(1-p)\cdot \log\binom{n+\E[P]}{n} \\
            &\le \log\binom{n+\E[P]}{n}-p\cdot \log\binom{n+\E[P]}{n}+2n.
        \end{align*}
        Since $\epsilon\ge n^2/U$, we have that $\E[P]=(U-n)\epsilon>n^2/2$, so $\log\binom{n+\E[P]}{n}\ge 100n$. Therefore,
        \begin{align*}
            \E[\text{cost}]&\le n\log\epsilon^{-1}+\E\Bk*{\log\binom{n+P}{n}}+R+1 \\
            &\le n\log\epsilon^{-1}+\log\binom{n+\E[P]}{n}-(1/10)\cdot 100n+4n \tag{$R\le n$} \\
            &<\log\binom{U}{n}.
        \end{align*}
    \end{itemize}
    In all cases, the expected cost of the protocol is less than $\log\binom{U}{n}$ bits, regardless of the distribution of $P$, which completes the proof of \cref{thm:non_augmented_lb_filter}.
\end{proof}

\section{Proof of Claim \ref{clm:NA_concentration}}
\label{app:NA_concentration}
In this appendix, we present a proof of \cref{clm:NA_concentration}.

We begin by recalling the definition of the variables $X_i$. In \cref{clm:perm_cover_large}, we consider a subset of $n/c - t$ rows (called \emph{selected rows}) among $n/c$ augmented rows in the matrix. The permutation entries (non-zero entries of the matrix) are determined by $\tperm$ independent permutations: The columns are partitioned into $n/c$ blocks, and for each row, each permutation activates exactly one column block according to the permutation values. 
Clearly, each permutation collectively activates $n/c -t$ column blocks in the selected rows.
The variable $X_i$ is defined to indicate whether the $i$-th column block is \emph{not} activated in \emph{any} of the selected rows. Given these definitions, \cref{clm:NA_concentration} can be restated as follows.

\NAConcentration*

To establish this result, we first demonstrate that the variables $X_i$ possess the negative association (NA) property. Our approach primarily utilizes techniques established in \cite{wajc2017negative}.

\begin{lemma}[\cite{wajc2017negative}]
    The following properties hold for negative association:
    \begin{enumerate}
        \item (Permutation rule.) Permutation distributions exhibit negative association. Specifically, if $X_1,\ldots, X_n$ are random variables that constitute a permutation of $v_1,\ldots,v_n$, then they are negatively associated.
        \item (Union rule.) The union of independent sets of negatively associated variables remains negatively associated.
        \item (Monotone function rule.) Independent monotone functions applied to negatively associated variables preserve negative association. Specifically, if $X_1,\ldots,X_n$ are negatively associated random variables, and $f_1,\ldots,f_m$ are all monotonically increasing (or all monotonically decreasing) functions that depend on disjoint variables, then $f_1(X),\ldots,f_m(X)$ are negatively associated.
    \end{enumerate}
\end{lemma}

\begin{claim}\label{lem:NA_randomcover}
The random variables $X_i$ defined in the proof of \cref{clm:perm_cover_large} are negatively associated.
\end{claim}

\begin{proof}
For each permutation $j \in [\tperm]$, we define auxiliary variables $H_{j,1}, \ldots, H_{j,n/c}$, where $H_{j,k}$ indicates whether the $j$-th permutation, restricted to the selected rows, fails to activate the $k$-th column block.
Then, for any fixed $j$, the variables $H_{j,1}, \ldots, H_{j,n/c}$ necessarily form a permutation of $t$ ones and $n/c-t$ zeros, thereby exhibiting negative association by the permutation rule. 

Since the permutations are sampled independently, the vectors $(H_{j,1}, \ldots, H_{j, n/c})$ are independent across different values of $j$. Consequently, by the union rule, the collection of all $H_{j,i}$ variables constitutes a set of negatively associated variables. Furthermore, since $X_i = \bigwedge_j H_{j,i}$ and the conjunction operation is monotonically decreasing, the monotone function rule establishes that $X_1, \ldots, X_{n/c}$ are negatively associated.
\end{proof}

For random variables exhibiting negative association, we can apply the following concentration inequality, which completes the proof of \cref{clm:NA_concentration}.

\begin{theorem}[Concentration of NA variables \cite{wajc2017negative}]\label{thm:NA_concentration}
    Let $X_1,\ldots,X_n$ be negatively associated random variables in $\{0,1\}$ and $\mu=\mathbb E\Bk*{\sum_{i=1}^{n} X_i}$. Then for any $k\geq 2$, the following inequality holds: $\Pr\Bk*{\sum_{i=1}^{n} X_i\geq k\mu}\leq (1/k)^{\Omega(k\mu)}$.
\end{theorem}

\end{document}